\documentclass[acmsmall,screen,nonacm]{acmart}

\usepackage{todonotes}
\usepackage{mathtools}
\usepackage{amsmath, amsthm, amsfonts}
\usepackage{prftree, mathpartir}
\usepackage{hyperref}
\usepackage[frozencache,cachedir=.]{minted}
\usepackage{stmaryrd}
\usepackage{thmtools, thm-restate}
\usepackage{caption}
\usepackage{subcaption}
\usepackage{wrapfig}
\usepackage{microtype}
\usepackage{tikz}
\usepackage{pgfplots}
\usepackage{tkz-euclide}
\usepackage{pgfplotstable}
\usetikzlibrary{calc}
\pgfplotsset{compat=1.18}
\usetikzlibrary{pgfplots.colorbrewer}
\usepgfplotslibrary{groupplots}

\def \tn #1[[#2]]{#1\llbracket #2\rrbracket}
\def \tnr #1[[#2]]{#1\left\llbracket #2\right\rrbracket^{\tt r}}
\def \rootpos {\varepsilon}
\def \bisim {\sim_\text{b}}
\def \sfork {\sim_\text{sf}}
\def \forkeq {\sim_\text{f}}
\def \calphaeq {\sim}
\def \lamdown {\downarrow}
\def \appleft {\swarrow}
\def \appright {\searrow}
\def \varup {\uparrow}
\def \gstep #1{\overset{#1}{\longrightarrow}}
\def \debruijn #1{\underline{#1}}

\def \valid #1{\mathbb P\!\left(#1\right)}
\def \vars #1{\mathbb V\!\left(#1\right)}
\def \free #1{\mathbb F\!\left(#1\right)}
\def \bound #1{\mathbb B\!\left(#1\right)}

\def \SCC {\mathop{\rm SCC}}

\def \lamsize #1{\left|#1\right|_\lambda}

\def \node#1{{\mathbf{#1}}}
\def \lift #1<#2>{#1\langle#2\rangle}
\def \gvar #1{{\rm g}(#1)}
\def \d#1{{\it #1}}

\def \glob{{\rm globalize}}
\def \globnaive{{\rm globalize_{naive}}}
\def \globscc{{\rm globalize_{scc}}}
\def \globstep{{\rm globalize_{step}}}

\theoremstyle{definition}
\newtheorem{definition}{Definition}[section]

\newtheorem{observation}[definition]{Observation}
\newtheorem{theorem}[definition]{Theorem}
\newtheorem{lemma}[definition]{Lemma}
\newtheorem{example}[definition]{Example}
\newtheorem{corollary}[definition]{Corollary}
\newtheorem{remark}[definition]{Remark}
\numberwithin{equation}{lemma}
\newenvironment{proofsketch}{%
  \renewcommand{\proofname}{proof sketch}\proof}{\endproof}

\setcopyright{rightsretained}
\acmYear{2024}
\copyrightyear{2024}

\begin{document}

\title{Hashing Modulo Context-Sensitive $\alpha$-Equivalence}
\author{Lasse Blaauwbroek}
\email{lasse@blaauwbroek.eu}
\author{Miroslav Ol\v s\'ak}
\email{mirek@olsak.net}
\affiliation{%
  \institution{Institut des Hautes Études Scientifiques}
  \city{Bures-sur-Yvette}
  \country{France}
}
\author{Herman Geuvers}
\email{herman@cs.ru.nl}
\affiliation{%
  \institution{Radboud University}
  \city{Nijmegen}
  \country{The Netherlands}
}

\begin{abstract}
  The notion of $\alpha$-equivalence between $\lambda$-terms is commonly used to
  identify terms that are considered equal. However, due to the primitive
  treatment of free variables, this notion falls short when comparing subterms
  occurring within a larger context. Depending on the usage of the Barendregt
  convention (choosing different variable names for all involved binders), it
  will equate either too few or too many subterms. We introduce a formal notion
  of context-sensitive $\alpha$-equivalence, where two open terms can be
  compared within a context that resolves their free variables. We show that
  this equivalence coincides exactly with the notion of bisimulation
  equivalence. Furthermore, we present an efficient $O(n\log n)$ runtime
  hashing scheme that identifies $\lambda$-terms modulo context-sensitive
  $\alpha$-equivalence, generalizing over traditional bisimulation partitioning
  algorithms and improving upon a previously established $O(n\log^2 n)$
  bound for a hashing modulo ordinary $\alpha$-equivalence by Maziarz et
  al~\cite{DBLP:conf/pldi/MaziarzELFJ21}. Hashing $\lambda$-terms is useful in
  many applications that require common subterm elimination and structure
  sharing. We have employed the algorithm to obtain a large-scale, densely packed,
  interconnected graph of mathematical knowledge from the Coq proof assistant
  for machine learning purposes.
\end{abstract}



\maketitle

\section{Introduction}
This paper studies equivalence of $\lambda$-terms modulo renaming of bound
variables, so called $\alpha$-equivalence. This has been studied extensively in
the history of $\lambda$-calculus, starting with Church~\cite{Church}. The
overview book by Barendregt~\cite{BarendregtLambdaCalc} also defines
and studies it in detail. There, $\alpha$-equivalence is defined as a relation
$t =_\alpha u$ between two $\lambda$-terms that captures the idea
that $t$ can be obtained from $u$ by renaming bound variables.

In the present paper we study a more general situation where $t$ and $u$ are
accompanied by a \textit{context} that binds their free variables. Hence, we
study the notion of {\em context-sensitive $\alpha$-equivalence}, which we will
show to coincide with bisimulation when interpreting $\lambda$-terms as graphs.
This notion has particular importance in case one wants to semantically compare
and de-duplicate the subterms of huge $\lambda$-terms (e.g. a dataset extracted
from the Coq proof assistant~\cite{the_coq_development_team_2020_4021912}). To
this end, we also define an efficient $O(n\log n)$-time hashing algorithm that
respects this equivalence, in the sense that sub-terms receive the same hash if
and only if they are context-sensitive $\alpha$-equivalent.

\subsection{Problem Description}

The $\alpha$-equivalence relation equates
$\lambda$-terms modulo the names of binders. For example, $\lambda x. \lambda y. x
=_\alpha \lambda y. \lambda x. y$. De Bruijn indices are a
way to make the syntax of $\lambda$-terms invariant w.r.t. the
$\alpha$-equivalence relation. Using de Bruijn indices, both terms above would
be represented by $\lambda\lambda\ \debruijn{1}$, and would hence be
syntactically equal.

The $\alpha$-equivalence relation becomes less clear when free variables are
involved. Usually, it is understood that terms are only equal when all
occurrences of free variables are equal. However, the situation is more complicated
when considering free variables within a known context. Consider the example
\begin{equation}\label{ex:false-negative}
  \lambda t.\ Q\ (\lambda z. \lambda f.\ f\ t)\ (\lambda g.\ g\ t).
\end{equation}
In this term, are the subterms $\lambda f.\ f\ t$ and $\lambda g.\ g\ t$
considered $\alpha$-equivalent? Most would agree they are equivalent because we can
share these terms with a let-construct without changing the meaning of the term:
\begin{equation}\label{ex:let-introduction}
\lambda t.\ \text{let}\ s \coloneqq \lambda f.\ f\ t\text{ in } Q\ (\lambda
z. s)\ s
\end{equation}
Here, the justification that we are ``not changing the meaning of the term'' is
that one $\beta$-reduction of the introduced \textit{let} will give us the
original term (modulo renaming of bound variables).
However, when we represent the original term with de Bruijn indices, these two
sub-terms are not syntactically equal.
\[\lambda\ Q\ (\lambda\lambda\ \debruijn{0}\ \debruijn{2})\ (\lambda\
  \debruijn{0}\ \debruijn{1})\]
The promise of de Bruijn indices has failed us! If we want to find the common
sub-terms of a program, we cannot simply convert the program to use de Bruijn
indices, hash the program into a Merkle tree, and call it a day. Other
representations, like de Bruijn levels, suffer from similar issues.

In addition to false negatives, de Bruijn indices can also lead to false
positives. Consider the example
\begin{equation}\label{ex:false-positive}
  \lambda t.\ Q\ (\lambda z. \lambda f.\ f\ z)\ (\lambda g.\ g\ t).
\end{equation}
Here, the subterms $\lambda f.\ f\ p$ and $\lambda g.\ g\ t$ are not
$\alpha$-equivalent. However, when expressed using de Bruijn indices
they become equal.
\[\lambda\ Q\ (\lambda\lambda\ \debruijn{0}\ \debruijn{1})\ (\lambda\
  \debruijn{0}\ \debruijn{1})\]

Given these counter-examples, one might conclude that de Bruijn indices are not
as useful in deciding equality between (sub)-terms as is commonly thought.
Unfortunately, the situation is not much better for $\lambda$-terms with
ordinary named variables. Take for example this naive attempt at defining
context-sensitive $\alpha$-equivalence:
\begin{quote}
  Two subterms of $t$ are $\alpha$-equivalent in the context of $t$ if the bound
  variables in $t$ can be renamed such that the subterms become syntactically equal.
\end{quote}
An immediate counter-example to this definition is the term
$Q\ (\lambda x.x)\ (\lambda xy.\ x)$
and the question whether the two occurrences of the variable $x$
are $\alpha$-equivalent. According to the definition, yes, but the variables
correspond to binders that cannot possibly be considered equivalent.

At this point, it is not even clear what precise equivalence relation we are
looking for, even though many people would ``know an equivalence between
subterms when they see one.'' The only intuitive idea we have to build on is the
introduction of a \textit{let}-abstraction in Formula~\ref{ex:let-introduction}.
But, as we will see, this is not sufficient on its own.

\subsection{Fork Equivalence}\label{sec:intro-fork}

Let us return to Example~\ref{ex:false-negative}, where we used
\textit{let}-abstraction to ``show'' that two subterms are equal. We make this
more precise (a fully formal treatment can be found in
Section~\ref{sec:definitions}). To conveniently talk about the equality of
subterms within a context, we underline the two terms of interest. This allows
us to restate the question in Example~\ref{ex:false-negative} as follows.
\begin{equation}\label{ex:false-negative-formal}
  \begin{gathered}
  \lambda t.\ Q\ (\lambda z. \underline{\lambda f.\ f\ t})\ (\underline{\lambda g.\ g\ t}) \\
  \lambda\ Q\ (\lambda\underline{\lambda\ \debruijn{0}\ \debruijn{2}})\ (\underline{\lambda\ \debruijn{0}\ \debruijn{1}})
  \end{gathered}
\end{equation}
The underlined subterms are the \textit{subject} of the context-sensitive
$\alpha$-equivalence. The remainder of the term, which we can write as $\lambda
t.\ Q\ (\lambda z.\ \_)\ \_$ or $\lambda\ Q\ (\lambda\ \_)\ \_$, is the
\textit{context} in which the equivalence is to be shown. Now, in order to
perform the \textit{let}-abstraction, we must split the context into two pieces,
between which we can insert the \textit{let}. In this case, the split we make is
$\lambda t.\ \_$ and $Q\ (\lambda z.\ \_)\ \_$. We can then show that the term
$\lambda f.\ f\ t$ is closed under the outer context and when we substitute it
into the holes in the inner context (while avoiding variable capture), we obtain
the original term. Hence, we can reassemble everything to obtain the same term
from Example~\ref{ex:let-introduction}.
\[
  \lambda t.\ \text{let}\ s \coloneqq \lambda f.\ f\ t\text{ in } Q\ (\lambda
  z. s)\ s
\]

\begin{figure}
  \centering
  \begin{subfigure}{.4\textwidth}
    \centering
    \includegraphics{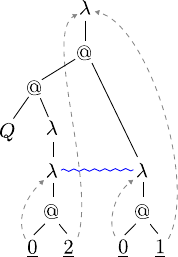}
    \caption{The single fork in Example~\ref{ex:false-negative-formal}.}
    \label{fig:fork}
  \end{subfigure}\hfill
  \begin{subfigure}{.6\textwidth}
    \centering
    \includegraphics{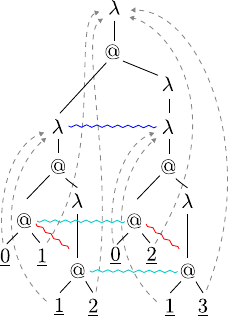}
    \caption{Illustration of sub-forks and transitivity from Example~\ref{ex:double-fork}.}
    \label{fig:double-fork}
  \end{subfigure}
  \caption{Illustrations of forks in terms built from de Bruijn indices. Equivalent
    sub-terms are related through a squiggly line. Back-edges from variables to
    binders are illustrative only.}
  \label{fig:test}
\end{figure}

\textit{Let}-abstractable subterms are the first ingredient of the
fork-equivalence relation\footnote{The formal definition we introduce later does
  not involve actual \textit{let}-expressions, but rather identifies the
  circumstances where a \textit{let}-abstraction can be performed.}. This
relation is named so because the relation is built upon a fork, starting with a
single outer context that forms the base of the fork. Then follows an inner context
that acts as the bifurcation of the base towards the two subjects. The fork is
shown in Figure~\ref{fig:fork}.

The definition of fork equivalence is not yet complete.
\textit{Let}-abstraction only allows us to compare subjects that share a common
context. However, when two subjects are closed, their context is irrelevant. In
that situation, both subjects may occur in a completely different context while
still being equivalent. This gives rise to an equivalence relation $\forkeq$
formed between two pairs of a subject and context. This allows establishing
equivalences between closed subjects such as
\[P\ (\underline{\lambda x.x}) \forkeq Q\ (\underline{\lambda y.y}).\]
In general, in addition to \textit{let}-abstraction, we say that a fork can also be
established between any two closed subjects that are equal modulo variable
renaming (or syntactically equal in the case of de Bruijn indices). Closed
subjects can be accompanied by arbitrary contexts, as they can never influence
the meaning of the subjects. \textit{Let}-abstraction, when phrased as a
relation $\forkeq$, would still require a common context even though that
common context is stated twice. Our running
Example~\ref{ex:false-negative-formal} would be phrased as
\[
  \lambda t.\ Q\ (\lambda z. \underline{\lambda f.\ f\ t})\ (\lambda g.\ g\ t)
  \forkeq
  \lambda t.\ Q\ (\lambda z. \lambda f.\ f\ t)\ (\underline{\lambda g.\ g\ t}).
\]

To complete the definition of fork equivalence, we must also extend it with both
the \textit{sub-fork} and with transitivity of forks. The following example
illustrates the need for these extensions:
\begin{equation}\label{ex:double-fork}
\begin{aligned}
  &\lambda t.\ (\lambda x.\ \underline{x\ t}\ (\lambda y.\ x\ t))\ (\lambda z.\ \lambda x.\ x\ t\ (\lambda y.\ x\ t))\\
  \forkeq {}&
  \lambda t.\ (\lambda x.\ x\ t\ (\lambda y.\ \underline{x\ t}))\ (\lambda z.\ \lambda x.\ x\ t\ (\lambda y.\ x\ t))\\
  \forkeq {}&
  \lambda t.\ (\lambda x.\ x\ t\ (\lambda y.\ x\ t))\ (\lambda z.\ \lambda x.\ x\ t\ (\lambda y.\ \underline{x\ t}))
\end{aligned}
\end{equation}
Notice how the first equivalence can be established because the subjects are
\textit{let}-abstractable. For the second equivalence, we have no such luck. There
is no place where we could split the context such that both instances of $x\ t$
would be closed under the outer context. Note, however, that if we widen the
subject $x\ t$ to $\lambda x.\ x\ t\ (\lambda y.\ x\ t)$, a
\textit{let}-abstraction can indeed be performed. Because both underscored instances of
$x\ t$ occur in the same position in this wider term, we allow for the creation
of a sub-fork between them. Finally, to demonstrate the transitivity concept, we
can combine both \textit{let}-abstractions in the following term, such that
all underscored locations coincide:
\[
\lambda t.\ \text{let }u \coloneqq (\lambda x.\ \text{let }v \coloneqq x\ t\text{ in }v\ \lambda y.v)\text{ in }u\ \lambda z. u
\]

This is further illustrated in Figure~\ref{fig:double-fork}. The two red forks
are created by \textit{let}-abstraction, as well as the dark blue fork. The two
teal lines are sub-forks of the dark blue fork. Finally, using
transitivity, fork equivalence is formed between all four connected sub-terms.

At this point, the definition of fork equivalence includes (1)
\textit{let}-abstraction, (2) equivalence between $\alpha$-equal closed terms, (3) the
sub-fork and (4) a transitive closure. It is not obvious that this relation is
indeed sound and complete w.r.t. the intuitive notion of equivalence between
subterms. To mitigate such doubts, we will now introduce a completely separate
equivalence relation that will be shown equal to fork equivalence in
Section~\ref{sec:proof-sketch}.

\subsection{Equivalence through Bisimulation}\label{sec:intro-bisim}

For our second equivalence relation, we will interpret $\lambda$-terms as
directed graphs. The skeleton of the graph is formed by the abstract syntax tree
of the $\lambda$-term, but instead of having variables or de Bruijn indices in the
leaves, they have a pointer back to the location where the variable was
bound.

\begin{figure}[b]
  \centering
  \begin{subfigure}{.6\textwidth}
    \centering
    \includegraphics{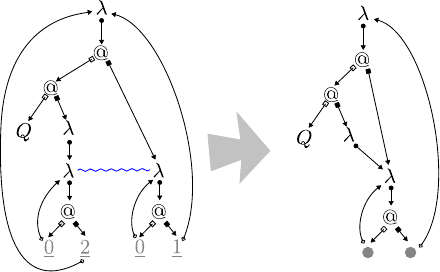}
    \caption{Subterms of Example~\ref{ex:false-negative} are bisimilar. Their
      nodes can be merged, resulting in variables without a unique de Bruijn
      index.}
    \label{fig:simple-bisim}
  \end{subfigure}\hfill
  \begin{subfigure}{.37\textwidth}
    \centering
    \includegraphics{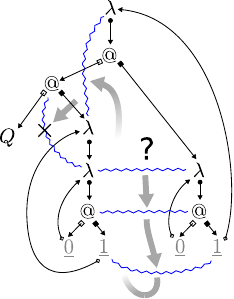}
    \caption{Subterms of Example~\ref{ex:false-positive} are not bisimilar,
      shown by a counter-example.}
    \label{fig:example-not-bisim}
  \end{subfigure}
  \caption{Illustrations of $\lambda$-terms as unordered graphs with labeled
    edges. Subject terms are related through a squiggly line. De Bruijn indices
    in variables are for illustrative purposes only.}
  \label{fig:bisim-examples}
\end{figure}

With such an interpretation, it becomes possible to define context-sensitive
$\alpha$-equivalence as the well-known notion of the \textit{bisimilarity relation}
that is common in the analysis of labeled transition systems and many other
(potentially) infinite structures~\cite{DBLP:books/daglib/0020348}. As a refresher, we restate
the definition of bisimilarity on a directed graph with labeled nodes and edges:
A relation $R$ between nodes is a bisimulation relation when for all nodes
$(p_1, p_2) \in R$ the labels of $p_1$ and $p_2$ are equal and
\begin{equation*}
  \begin{split}
    &\text{if}\ p_1\xrightarrow{a} q_1\ \text{then there exists}\ q_2\ \text{such that}\ p_2\xrightarrow{a}q_2\ \text{and}\ (q_1, q_2) \in R\\
    &\text{if}\ p_2\xrightarrow{a} q_2\ \text{then there exists}\ q_1\ \text{such that}\ p_1\xrightarrow{a}q_1\ \text{and}\ (q_1, q_2) \in R.
  \end{split}
\end{equation*}
Two nodes $p_1$ and $p_2$ in a graph are then considered bisimilar if there
exists a bisimulation relation $R$ such that $(p_1, p_2)\in R$. The bisimilarity
relation is the union of all bisimulation relations, and is itself a
bisimulation.

Figure~\ref{fig:simple-bisim} shows the term in Example~\ref{ex:false-negative}
as a graph. Due to the unordered nature of graphs, in contrast to more common
presentation of syntax trees as ordered trees, edges are annotated with a shape
that represents their label. It is straight-forward to build a relation $R$ that
includes the two subject terms and satisfies the requirements of a bisimulation
relation. The existence of the bisimulation certifies that we can de-duplicate
the two subject terms in the graph structure without changing their meaning.
Note that as a result, variables can no longer be assigned a unique de Bruijn
index.

Figure~\ref{fig:example-not-bisim} shows how the subterms in
Example~\ref{ex:false-positive} are not bisimilar. This is done by assuming a
valid bisimulation relation $R$ that contains the two subject terms. One can
then simultaneously follow equal edges on both sides until two nodes with
different labels are encountered that are not allowed to be in $R$ (marked with
a cross in the figure).

\begin{wrapfigure}[18]{r}{0.37\textwidth}
  \vspace{-1em}
  \centerline{\includegraphics{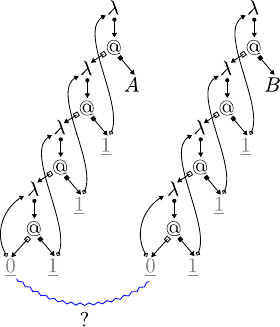}}
  \caption{To decide whether two variables are bisimilar, we must examine
    bisimilarity between far-away terms $A$ and $B$.}
  \label{fig:many-jumps}
\end{wrapfigure}
Deciding whether two subterms are equal is not always as easy as the
illustrations in Figure~\ref{fig:bisim-examples} suggest. Consider the
bisimulation question in Figure~\ref{fig:many-jumps} between two variables. To
decide whether the variables are bisimilar, we must repeatedly travel up and down
into the term jumping between variables and their binders, until we reach the
root. From there, we must travel into the subterms $A$ and $B$. The two
variables will be bisimilar if and only if $A$ and $B$ are bisimilar. This shows
that no matter the size of the subject terms, one might need to traverse and
examine the entire context in order to decide their bisimilarity. A decision
procedure or hashing scheme must take all of this information into account while
still being efficient.

Finally, we note the importance of having variable nodes in the
embedding of $\lambda$-terms as graphs. Since variable nodes only have a single
incoming edge and a single outgoing back-edge to the corresponding binder, one
might be tempted to skip the variable node altogether. However, such an
embedding would lead to a situation where we share too many terms. One example
of such problems is the following sequence of terms that should not be bisimilar.
\[
  \lambda x.x \bisim \lambda y.\lambda x.x \bisim \lambda
  z.\lambda y.\lambda x.x \bisim \ldots
\]
However, under the following encoding, which omits variables, these terms would
be bisimilar.
\begin{center}
  \includegraphics{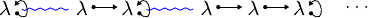}
\end{center}

\begin{wrapfigure}[8]{r}{0.27\textwidth}
  \centering
  \begin{tikzpicture}
    \tkzDefPoint(0,0){O};
    \tkzDefPoint(35:1){A};
    \tkzDefPoint(185:1){B};
    \tkzDefPoint(315:1){C};
    \tkzDrawArc[stealth-, delta=-15](O,A)(B);
    \tkzDrawArc[stealth-, delta=-15](O,B)(C);
    \tkzDrawArc[stealth-, delta=-15](O,C)(A);
    \node at (A) {hash algorithm};
    \node at ($ (B) + (0.5,0) $) {fork equivalence};
    \node at (C) {bisimilarity};
  \end{tikzpicture}
  \caption{Proof strategy.}
  \label{fig:proof-circular}
\end{wrapfigure}
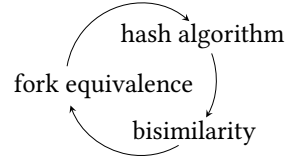
Fork equivalence and bisimilarity are defined quite differently. Indeed, it
might be surprising that they end up being exactly equal. These two dissimilar
characterizations of context-sensitive $\alpha$-equivalence are crucial
in the correctness proof of the hashing algorithm we present in
Section~\ref{sec:algorithms}. Figure~\ref{fig:proof-circular} outlines how we
will prove the equality between the algorithm and the two equivalences.
Section~\ref{sec:proof-sketch} will show that if two subterms are fork-equivalent,
they will receive the same hash. Conversely, if two subterms receive the same
hash, they must be bisimilar. The first direction is considerably simplified
because fork equivalence is characterized as a sequence of syntactically simple
single forks. The other direction can be proven by showing that every
$\lambda$-term is bisimilar to its hashed version (given an appropriate
extension of bisimulation).

\subsection{Hashing versus Partitioning Modulo Bisimulation}
\label{sec:versus-bisim}

Given that $\lambda$-terms can be expressed as a deterministic labeled
transition systems such that context-sensitive $\alpha$-equivalence corresponds
to bisimilarity, we can draw apply the large body of research in this area to
$\lambda$-terms. Indeed, there are off-the-shelf algorithms that partition a
deterministic graph $G = (E, V)$ modulo bisimulation in $O(|E|\log |V|)$
time~\cite{DBLP:journals/siamcomp/PaigeT87, DBLP:journals/tcs/DovierPP04,
  hopcroft1971n, DBLP:journals/ipl/Valmari12}. For a $\lambda$-term of size $n$,
this translates into a $O(n\log n)$ runtime.

The algorithm we present in Section~\ref{sec:algorithms} does not partition a
$\lambda$-term but rather assigns a hash-code to each position in the term. If
no hash-collisions occur, two positions are context-sensitively
$\alpha$-equivalent if and only if their hash is equal. If desired, collisions can be
detected by combining our algorithm with
hash-consing~\cite{DBLP:conf/ml/FilliatreC06}. As such, our algorithm is
strictly more powerful than a partitioning algorithm. The crucial advantage is
that $\lambda$-terms may be hashed independently of each other while the
resulting hashes can still be meaningfully compared. This enables one to
efficiently find duplicate terms as they are created by inserting their hashes
in a table. Partitioning schemes do not allow this, because all terms must
be simultaneously input into the algorithm. We are not aware of
any existing graph-hashing algorithms that operate modulo bisimilarity.

Our hashing scheme is straight-forward to implement compared to general-purpose
partitioning algorithms, and relies only on well-known $\lambda$-calculus
operations such as capture-avoiding substitution.
Section~\ref{sec:experimental-evaluation} shows it can outperform
partitioning algorithms in all realistic scenarios.

\subsection{Context-Sensitive $\alpha$-Equivalence versus Ordinary
  $\alpha$-Equivalence}
\label{sec:versus}

We will now explore some of the differences between context-sensitive
$\alpha$-equivalence and ordinary $\alpha$-equivalence. Previously, Maziarz et
al~\cite{DBLP:conf/pldi/MaziarzELFJ21} have constructed a hashing scheme modulo
ordinary $\alpha$-equivalence. How does this differ from our scheme and when is
one relation or algorithm to be preferred over another?

To examine this question, we again look at the term in
Example~\ref{ex:double-fork} as visualized in Figure~\ref{fig:double-fork}.
\begin{align*}
  \lambda t.\ (\lambda x.\ x\ t\ (\lambda y.\ x\ t))\ (\lambda z.\
  \lambda x.\ x\ t\ (\lambda y.\ x\ t)) \\
  \lambda t.\ (\lambda x.\ \debruijn 0\ \debruijn 1\ (\lambda y.\ \debruijn 1\ \debruijn 2))\ (\lambda z.\
  \lambda x.\ \debruijn 0\ \debruijn 2\ (\lambda y.\ \debruijn 1\ \debruijn 3))
\end{align*}
This term contains four instances of the term $(x\ t)$, all of whom are
represented differently using de Bruijn indices, and all of whom are considered
equal modulo context-sensitive $\alpha$-equivalence. Which of these instances of
$x\ t$ are equal under ordinary $\alpha$-equivalence? This is a tricky question
to answer. Both the variables $x$ and $t$ are free in the term $(x\ t)$. As
such, under ordinary $\alpha$-equivalence, they would be compared syntactically.
This would indeed make all four instances $\alpha$-equivalent. However, such an
interpretation would get us into trouble when we rename the bound variables in
the term. This may change the variable names in the subterms, which in turn may
cause them to no longer be $\alpha$-equivalent. Hence, $\alpha$-equality between
subterms is contingent on the particular choice of bound variable names we have
made. That is not acceptable. Maziarz et al. solve this ambiguity by globally
enforcing the Barendregt convention on the entire universe of $\lambda$-terms.
That is, no two binders may ever have the same name. The term in our example
does not satisfy the Barendregt convention. We must rewrite it to
\[
  \lambda t.\ (\lambda x_1.\ x_1\ t\ (\lambda y_1.\ x_1\ t))\ (\lambda z.\
  \lambda x_2.\ x_2\ t\ (\lambda y_2.\ x_2\ t)).
\]
We now have two $\alpha$-equivalent instances of $(x_1\ t)$ and two
$\alpha$-equivalent instances of $(x_2\ t)$. This is contrasted by
context-sensitive $\alpha$-equivalence, where all four terms are equal. In general,
we have that ordinary $\alpha$-equivalence is a sub-relation of
context-sensitive $\alpha$-equivalence\footnote{A direct comparison of the two
  relations is impossible because one is defined on the domain of
  $\lambda$-terms, while the other is defined on $\lambda$-terms with a context.
  However, one can imagine a trivial lift of ordinary $\alpha$-equivalence to
  $\lambda$-terms with a context, such that the context is entirely ignored.}.
For ordinary $\alpha$-equivalence, this can lead to some counter-intuitive situations: The subterms $\lambda
x_1.\ x_1\ t\ (\lambda y_1.\ x_1\ t)$ and $\lambda x_2.\ x_2\ t\ (\lambda y_2.\
x_2\ t)$ are considered $\alpha$-equivalent according to Maziarz et al., but
not all their constituent parts are mutually $\alpha$-equivalent. No such issues
exist for the context-sensitive variant.

Trade-offs between the two relations are as follows:
\begin{itemize}
\item \sloppy Ordinary $\alpha$-equivalence is simple. It is defined on $\lambda$-terms,
  while context-sensitive $\alpha$-equivalence needs an additional context.
\item Ordinary $\alpha$-equivalence cannot be defined properly on open terms
  encoded with de Bruijn indices. Syntactic comparisons between open de Bruijn
  indices leads to incorrect results (see Example~\ref{ex:false-positive}). A hybrid approach, such as a
  locally-nameless representation~\cite{DBLP:journals/jar/Chargueraud12}, is
  required to resolve this. Context-sensitive $\alpha$-equivalence can be
  properly defined for any representation of $\lambda$-terms.
\item When interpreting $\lambda$-terms as a graph, one should use
  context-sensitive $\alpha$-equivalence because it equals the graph-theoretic
  notion of bisimilarity.
\item For tasks like common subterm elimination in compilers, both relations may
  be appropriate because there one usually seeks to find the largest
  $\alpha$-equivalent subterms and not their descendants. Both relations achieve
  this.
\end{itemize}
Furthermore, the trade-offs between the hashing algorithm in this paper and the
one by Maziarz et al. are as follows:
\begin{itemize}
\item Our algorithm hashes all nodes in a term in $O(n\log n)$ time while
  Maziarz et al. require $O(n\log^2 n)$ time.
\item Maziarz et al. require a global preprocessing step to enforce the
  Barendregt convention. No such step is required for our algorithm.
\item The algorithm by Maziarz et al. is compositional. That is, given two
  hashed terms $t$ and $u$, one can derive the hash for
  $(t\ u)$ from the hash of the children. This may be a desirable property in a compiler, allowing one to maintain
  correct hashes while rewriting terms during optimization
  passes\footnote{Re-hashing a term after rewriting a term at depth $h$ may
    still be expensive, taking up to $O(h^2)$ time.}.
  Compositionality is not possible for context-sensitive $\alpha$-equivalence,
  because changing the context may require a change in the hash of all subterms
  (see Figure~\ref{fig:many-jumps}).
\item Maziarz et al. rely fundamentally on named variables, while we rely
  fundamentally on de Bruijn indices. In both cases, it is possible to do a
  representation conversion before hashing, but both algorithms have a clear ``native''
  representation. Much has been said around the relative merits of
  $\lambda$-term representations~\cite{DBLP:journals/entcs/BerghoferU07}, and we
  do not wish to make a value judgement here. It is good to know there are viable
  algorithms for both representational approaches, even if those algorithms have
  subtle differences in how they operate.
\end{itemize}

\subsection{Applications}
The original motivation for the subterm sharing algorithm in this paper was the
creation of a large-scale, graph-based machine learning dataset of terms in the
calculus of inductive constructions, exported from the Coq Proof
Assistant~\cite{the_coq_development_team_2020_4021912}. Extracting data from
over a hundred Coq developments leads to a single, interconnected graph
containing over 520k definitions, theorems and proofs. For each node in this
graph we calculate a hash modulo context-sensitive $\alpha$-equivalence. The
hash is then used to maximally share all subterms, resulting in a dense graph
with approximately 250 million nodes. A very small section of this graph is
visualized in Figure~\ref{fig:web}. More details on the construction can be
found elsewhere~\cite{web-paper}. Experiments show that subterm
sharing allows for an 88\% reduction in the number of nodes. Hence, without
sharing, such a graph would have over 2 billion nodes. Identifying identical
subterms in a graph is helpful semantic information that can be used by machine
learning algorithms to make predictions. The graph dataset has been leveraged to
train a state of the art graph neural network to synthesize proofs in
Coq~\cite{graph2tac}.
\begin{figure}
  \centering
  \includegraphics[width=\textwidth]{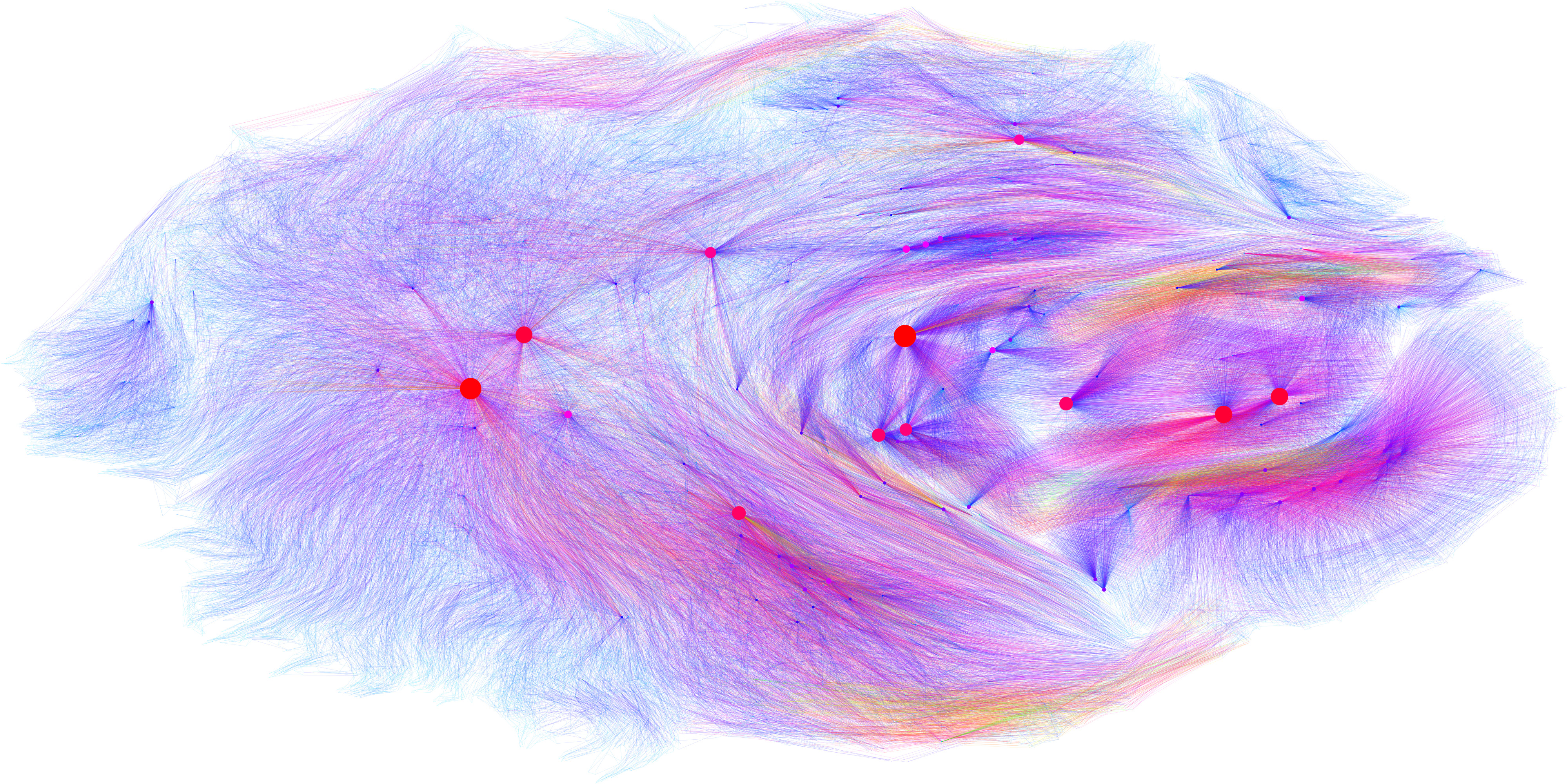}
  \caption{A visualization of a maximally shared graph of CIC terms extracted from
    Coq's Prelude.}
  \label{fig:web}
\end{figure}

In addition to our original motivation, we expect our algorithm to be useful in
other applications as well. In compilers, common sub-expression elimination is a
common optimization pass~\cite{DBLP:conf/comop/Cocke70} that can be performed
quickly using our algorithm. This was the original motivation for the algorithm
by Maziarz et al.~\cite{DBLP:conf/pldi/MaziarzELFJ21}. Hashes can also be used
by content addressable programming languages like
Unison~\cite{Chiusano_Bjarnason_Irani} and to build indexes of program libraries
that can then be queried to find opportunities for code
sharing~\cite{DBLP:conf/iwsc/ThomsenH12} or plagiarism
detection~\cite{DBLP:conf/bwcca/ZhaoXFC15}.

\subsection{Contributions}

In this paper, we develop a notion of context-sensitive $\alpha$-equivalence
that compares potentially open terms within a context. We define this notion
both through fork equivalence and bisimulation, and prove that these
approaches are equivalent. Note that we are not the first to study bisimilarity
on the syntax of $\lambda$-terms. In the context of optimal sharing for
efficient evaluation this is common~\cite{DBLP:conf/ppdp/CondoluciAC19,
  DBLP:conf/icfp/GrabmayerR14}. However, our relational characterization between
subterms with a context appears to be novel.
We
present an efficient decision procedure and hashing algorithm that identifies
terms modulo context-sensitive $\alpha$-equivalence. These algorithms have been
successfully used to efficiently encode a graph of Coq terms for machine
learning purposes. A reference implementation written in OCaml is
available~\cite{lasse_blaauwbroek_2023_10421517}.

The remainder of the paper is organized as follows.
Section~\ref{sec:definitions} first introduces the preliminaries followed by the
formal definitions of fork equivalence in Section~\ref{sec:fork-equivalence}
and bisimilarity in Section~\ref{sec:bisimulation-relation}. Then,
Section~\ref{sec:naive-alg} presents a simple, naive algorithm for hashing
$\lambda$-terms in $O(n^2)$ time. This is then refined in
Section~\ref{sec:alg-efficient} and finally a concrete, $O(n\log n)$ hashing
algorithm implemented in OCaml is presented in Section~\ref{sec:alg-concrete}.
Proofs of equality between the two
relations and the algorithms can be found in Section~\ref{sec:proof-sketch}.

\section{Definitions}
\label{sec:definitions}

In this section, we will further formalize the equivalence relations presented
in the introduction. From this point, we will only consider $\lambda$-terms
encoded with de Bruijn indices. The algorithms rely heavily on the fact that two
$\alpha$-equivalent closed terms encoded using de Bruijn indices are
syntactically equal. That said, the equivalence relations and the proofs of
equality between them also work when one uses $\lambda$-terms with named
variables modulo ordinary $\alpha$-equivalence. For the sake of legibility, we
will frequently still give examples using $\lambda$-terms with named
variables.

\subsection{Terms, Positions and Indexing}
\begin{definition}[$\lambda$-terms]
$\lambda$-terms with de Bruijn indices are generated
by the grammar
\[ t \Coloneqq \underline{i} \mid t\ t \mid \lambda\ t, \]
where a de Bruijn index $\debruijn i$ is a nonnegative integer.
We denote $t[\debruijn i\coloneqq u]$ to be the well-known capture-avoiding substitution of
variable $\debruijn i$ by $u$ in $t$. Furthermore, let
$\sigma = [u_0, u_1,\ldots, u_{n-1}]$ be a list of terms. Then $t\sigma$ denotes
the simultaneous subtitution of all variables $\debruijn i$ by $u_i$ in $t$.
\end{definition}

\begin{definition}[term indexing]
Let $p$ be a \textit{term position} generated by the grammar $\{\lamdown,\appleft, \appright\}^*$.
The indexing operation $t[p]$ is a partial function defined by
\begin{mathpar}
  (t\ u)[\appleft p] = t[p] \and
  (t\ u)[\appright p] = u[p] \and
  (\lambda\ t)[\lamdown p] = t[p] \and
  t[\rootpos] = t.
\end{mathpar}
\end{definition}
\begin{example}
  Consider a term $t = \lambda\ (\lambda\ A\ B)\ (\lambda\ C\ D)$ and positions $p =\ \lamdown\appleft$ and $q
  =\ \lamdown\appright$. Then
\begin{mathpar}
  t[\rootpos] = \lambda\ (\lambda\ A\ B)\ (\lambda\ C\ D)
  \and t[pp] = A \and t[pq] = B \and t[qp] = C \and t[qq] = D.
\end{mathpar}
\end{example}

\begin{definition}[position sets]
Let $\lamsize{p}$ denote the number of $\lamdown$s in a term position $p$. We
define the following sets of positions for a $\lambda$-term $t$.
\begin{align*}
  \valid{t} &= \{p \mid t[p] \text{ is defined}\}\quad && \text{The set of all valid positions in $t$.} \\
  \vars{t} &= \{ p \in \valid{t} \mid t[p] = \debruijn i\}\quad && \text{Positions that represent a variable in $t$.} \\
  \bound{t} &= \{ p \in \vars{t} \mid t[p] < \debruijn{\lamsize{p}} \} && \text{Positions that represent a bound variable.} \\
  \free{t} &= \{p \in \vars{t} \mid t[p] \geq \debruijn{\lamsize{p}} \} && \text{Positions that represent a free variable.}
\end{align*}
\end{definition}

\subsection{Locally Closed Subterms}
A subterm rooted at position $p$ in $t$ is considered to be \d{locally closed} in
$t$ if all of the free variables in $t[p]$ are also free in $t$. We will later use
this concept while formalizing \textit{let}-abstraction.
\begin{definition}
  Position $p\in\valid{t}$ is \d{locally closed} in $t$ if for every
  $q\in\free{t[p]}$ we have $pq\in\free{t}$.
\end{definition}
\begin{example}
  Consider again Example~\ref{ex:false-negative} as visualized in Figure~\ref{fig:fork}:
  \[u = \lambda t.\ Q\ (\lambda z. \lambda f.\ f\ t)\ (\lambda g.\ g\ t).\] Let $p =
  \ \lamdown\appleft\appright$, giving us $u[p] = \lambda z. \lambda f.\ f\ t$ and
  $u[p\lamdown] = \lambda f.\ f\ t$. The position $p\lamdown$ is locally closed
  in $u[p]$, because the variable $f$ is bound in $u[p\lamdown]$, and $t$ is
  free in both $u[p]$ and $u[p\lamdown]$. On the other hand, $p\lamdown$ is
  not locally closed in $u$, because variable $t$ is free in $u[p\lamdown]$ but bound in $u$.
\end{example}

Using this notion, we can introduce a more semantic version of term indexing.
Intuitively, $\lift t<p>$ will denote the subterm $t[p]$ with its de Bruijn
indices shifted to skip the context given by $p$. This way, we ``lift'' $t[p]$
out of its context. This can only be done if $p$ is locally closed in $t$.

\begin{definition}
  \label{def:lift}
  For $p \in \valid{t}$, define $\lift t<p>$ to be equal to
  $t[p]$, except that for all $q\in\free{t[p]}$ we have $\lift t<p>[q] = t[pq]
    - \lamsize{p}$.
\end{definition}
The term $\lift t<p>$ is a valid term (with non-negative
indices) only when $p$ is locally closed in $t$.
The operation $\lift t<p>$ is crucial in much of the analysis below, as it
allows us to abstract away from manually doing arithmetic on de Bruijn indices.
Another natural view of $\lift t<p>$ is that it reverses capture-avoiding
substitution, as shown in the following observation.

\begin{observation}
  Let $t$ be a $\lambda$-term with a free variable $\debruijn i$ at position
  $p$. That is, $p\in\free{t}$ such that $t[p] = \debruijn i$. Then $\lift
  t[\debruijn i \coloneq u]<p> = u$.
\end{observation}
Note that $t[\debruijn i \coloneq u][p] = u$ does not hold, because
the capture-avoiding substitution may have shifted the free de Bruijn
indices in $u$. The semantic indexing operation reverses these shifts.

\begin{example}
  Consider the term $t = \lambda\ \debruijn 0\ (\lambda\ \debruijn 0\ \debruijn
  2)$, where $\debruijn 2$ is a free variable. The position $\lamdown\appright$
  is locally closed in $t$ because the outer-most $\lambda$ is not referenced.
  We then have $\lift t<\lamdown\appright> = \lambda\ \debruijn 0\ \debruijn
  1$ compared to $t[\lamdown\appright] = \lambda\ \debruijn 0\ \debruijn 2$.
  On the other hand, neither $\lamdown$ nor $\lamdown\appleft$ are locally
  closed positions.
\end{example}

\subsection{Term Nodes}
We now formally define the notion of a term with a context as introduced in
Section~\ref{sec:intro-fork}.

\begin{definition}[term node]
  \label{def:term-node}
  Let $\tn t[[p]]$ denote a pair $(t, p)$ such that $t$ is a closed term,
  and $p\in\valid{t}$.
\end{definition}

In a pair $\tn t[[p]]$, the term $t[p]$ represents the \textit{subject}, while
the remaining part of $t$ is the \textit{context}.
Intuitively, we can think of $\tn t[[p]]$ as the subterm
at $t[p]$ but without losing the information about the context.

We call a pair $\tn t[[p]]$ a
\textit{term node} because it represents a node in the graph induced on
$\lambda$-terms that we introduced in Section~\ref{sec:intro-bisim}.
In the graph visualizations in Section~\ref{sec:intro-bisim}, each node is labeled only with the
top-most symbol of the subject term, either $\lamdown$, $@$, or $\varup$.
Although this is convenient from a visualization perspective, it does not work
from a mathematical perspective where a graph is defined as a pair of sets $G = (V, E)$ that
determine the nodes
and edges.
Taking the set of nodes to be $V = \{\lamdown, @, \varup\}$ would give us
trivial graphs with three nodes. Instead, a node $n \in V$ must uniquely
represent the subject term and the context in which it occurs. This is exactly
what a term node is.

In order to formally define the graph on $\lambda$-terms, we must also define
the set $E$ of transitions between term nodes.

\begin{definition}[term node transitions]
  \label{def:node-transitions}

We define the transitions between term nodes as follows. Let $\tn t[[p]]$ be a
term node, and $x \in \{\lamdown, \appleft, \appright\}$ such that $x\in
\valid{t[p]}$. Then,
\[\tn t[[p]] \gstep x \tn t[[px]].\]
In addition, a term node whose subject is a variable also has a transition to
the corresponding binder. Formally, if $q\lamdown r \in \vars{t}$ and
$t[q\lamdown r] = \debruijn{\lamsize{r}}$, then
\[
\tn t[[q\lamdown r]] \gstep\varup \tn t[[q]].
\]
\end{definition}

Note that for a term node $\tn t[[p]]$ such that $p \in \vars{t}$, we can always
make a split $p = q\lamdown r$ such that $t[p] = \debruijn{\lamsize{r}}$. This
is because the definition of a term node $\tn t[[p]]$ stipulates that $t$
is closed and hence $\lamsize{p} > t[p]$.




Now, we have all the tools needed to formally specify the two equal notions of
context-sensitive $\alpha$-equivalences outlined in the introduction.

\subsection{Fork Equivalence}
\label{sec:fork-equivalence}
Here, we formalize the concepts introduced in Section~\ref{sec:intro-fork},
starting with the notion of a single fork.

\begin{definition}[single fork]
  \label{def:single-fork}
  A \textit{single fork} between term nodes, denoted by $\tn
  t_1[[p_1]]\sfork\tn t_2[[p_2]]$, exists when one of the following rules is
  satisfied.
  \begin{mathpar}
    \prftree[r]{\textit{let}-abs}
    {q_1 \text{ locally closed in } t[p]}
    {\lift t[p]<q_1> = \lift t[p]<q_2>}
    {\tn t[[pq_1r]] \sfork \tn t[[pq_2r]]}
    \and
    \prftree[r]{closed}
    {t_1[p_1] \text{ closed}}
    {t_1[p_1] = t_2[p_2]}
    {\tn t_1[[p_1r]] \sfork \tn t_2[[p_2r]]}
  \end{mathpar}
  It is assumed that $\tn t[[pq_1r]]$, $\tn t[[pq_2r]]$, $\tn t_1[[p_1r]]$ and
  $\tn t_2[[p_2r]]$ are valid term nodes in the rules above.
\end{definition}
When $\tn t[[pq_1r]]\sfork\tn t[[pq_2r]]$ is satisfied by the \textit{let}-abs
rule, this means that a \textit{let} can be introduced in $t$ at position $p$.
The \textit{let} binds the term $\lift t[p]<q_1>$, and the terms at position $pq_1$ and
$pq_2$ can be changed into a variable pointing to the \textit{let}.
Example~\ref{ex:let-introduction} illustrates this. The rule for closed terms is
simpler. It states that a closed term is equivalent to itself in an arbitrary
context. Finally, note that the conclusion of both rules allow for an arbitrary
position $r$ that ``extends'' the prongs of a known fork, as illustrated in the
following observation:

\begin{observation}[sub-fork]
  \label{obs:forkprop}
  $\tn t_1[[p_1]] \sfork \tn t_2[[p_2]]$ implies
  $\tn t_1[[p_1r]] \sfork \tn t_2[[p_2r]]$ for $r\in\valid{t_1[p_1]}$.
\end{observation}

The relation for a single fork is reflexive. For every term node $\tn t[[p]]$,
a self-fork can be constructed using the \textit{let}-abs rule by taking $q_1 =
q_2 = \rootpos$. In addition, a fork is symmetric. Transitivity does not hold,
however. To obtain an equivalence, we take the transitive closure.

\begin{definition}[fork equivalence] $\tn t_1[[p_1]] \forkeq \tn
  t_2[[p_2]]$ is the transitive closure of $\tn t_1[[p_1]] \sfork \tn t_2[[p_2]]$.
\end{definition}

\subsection{Bisimilarity}
\label{sec:bisimulation-relation}

In addition to fork equivalence, we also formalize the bisimulation relation
introduced in Section~\ref{sec:intro-bisim}.

\begin{definition}[bisimilarity]
  \label{def:bisim}
  A binary relation $R$ on term nodes is called a \d{bisimulation} if for every
  pair of term nodes $(n_1, n_2)\in R$ and every $x\in
  \{\lamdown,\appleft,\appright,\varup\}$ the following holds:
  \begin{itemize}
  \item If $n_1 \gstep x n_1'$, then there
    exists $n_2'$ such that $n_2 \gstep x n_2'$ and $(n_1', n_2') \in R$.
  \item If $n_2 \gstep x n_2'$, then there
    exists $n_1'$ such that $n_1 \gstep x n_1'$ and $(n_1', n_2') \in R$.
  \end{itemize}
  Two term nodes $\tn t_1[[p_1]]$ and $\tn t_2[[p_2]]$ are \textit{bisimilar}, denoted
  $\tn t_1[[p_1]]\bisim\tn t_2[[p_2]]$, if there exists a bisimulation $R$ such
  that $(\tn t_1[[p_1]], \tn t_2[[p_2]])\in R$.
\end{definition}

It is well-known that bisimilarity is an equivalence relation, and that it is
itself a bisimulation. Note that unlike the informal definition in
Section~\ref{sec:intro-bisim}, this bisimulation relation does not directly
compare the labels of term nodes (the label of a node $\tn t[[p]]$ would be the
root symbol of $t[p]$). This is not needed, because the label of a node is fully
determined by the labels of its outgoing edges. This, together with the fact
that the transition system is deterministic, considerably simplifies our setup
compared to arbitrary transition systems. In particular, the notion of
\textit{bisimulation} coincides with the notion of \textit{simulation}, in which
the second clause of Definition~\ref{def:bisim} is omitted.

One of the main results of this paper, proved in Section~\ref{sec:proof-sketch}, is
that fork equivalence and bisimilarity are equal.
\begin{theorem}
  \label{thm:fork-bisim}
  $\tn t_1[[p_1]] \bisim \tn t_2[[p_2]]$ if and only if
  $\tn t_1[[p_1]] \forkeq \tn t_2[[p_2]]$.
\end{theorem}

\begin{definition}[context-sensitive $\alpha$-equivalence]
  Since the main equivalence notion is captured both by $\forkeq$ and $\bisim$,
  we will use $\tn t_1[[p_1]] \calphaeq \tn t_2[[p_2]]$ to denote context-sensitive $\alpha$-equivalence and
  switch between the two interpretations at will.
\end{definition}

\section{Deciding Context-Sensitive
  $\alpha$-Equivalence Through Globalization}\label{sec:algorithms}
As we saw in the introduction, using syntactic equality on $\lambda$-terms with
de Bruijn indices is problematic in the presence of free variables. Such
variables need to be interpreted within a context in order to be meaningful. Our
approach to deciding whether or not two terms are $\alpha$-equivalent in a given
context is to \textit{globalize} the variables. We replace all de Bruijn indices
in a term with \textit{global variables}, which are structures that contain
exactly the required information to capture the context that is relevant to the
variable. After globalization, we can indeed compare two subterms syntactically
without having to consider the context in which they exist, because that context
has been internalized into the variables.

As it happens, the structure associated to a global variable is itself a
$\lambda$-term that may contain de Bruijn indices or further global variables.
This leads us to extend the grammar of $\lambda$-terms into that of $g$-terms.

\begin{definition}[$g$-terms] A $g$-term is generated by the grammar
  \[ t \Coloneqq \debruijn{i} \mid t\ t \mid \lambda\ t \mid \gvar t\] where a
  term of the form $\gvar t$ represents a \d{global variable} labeled by a
  $g$-term $t$. We consider any $\lambda$-term to also be a $g$-term, and
  trivially lift all operations and relations defined on $\lambda$-terms:
  \begin{itemize}
  \item Substitution is extended such that $\gvar{t}[\debruijn i
    \coloneq u] = \gvar{t}$, without traversing into $t$.
  \item Term indexing $t[p]$ behaves identical to $\lambda$-terms. Indexing does
    not extend into the structure of a global variable. The functions $\lift
    t<p>$, $\valid{t}$, $\vars{t}$, $\free{t}$ and $\bound{t})$ remain defined
    as before. Global variables are not part of the set $\vars{t}$. A global
    variable $\gvar{u}$ is always closed.
  \item The definition for $\forkeq$ remains as before. The definition of
    $\bisim$ is extended, but we postpone this until Section~\ref{sec:proof-sketch}.
  \end{itemize}
\end{definition}

We are now ready to describe algorithms that transform a $\lambda$-term into a
$g$-term that respects context-sensitive $\alpha$-equivalence. We start with a
slow, naive algorithm which is then made more efficient. Finally, we present a
practical OCaml program that processes a term and attaches an appropriate hash
to every subterm.

\subsection{A Naive Globalization Procedure}
\label{sec:naive-alg}
Contrary to the relation $\tn t_1[[p_1]]\calphaeq \tn t_2[[p_2]]$, the
globalization procedure does not operate on term nodes but rather on
closed $g$-terms. This works, because closed terms do not require a context.

\begin{definition}[naive globalization] Recursively define $\globnaive(t)$ from
  closed $g$-terms to closed $g$-terms as follows.
\begin{alignat*}{2}
  &\globnaive(\lambda\ t)&&=\lambda\ \globnaive(t[\debruijn 0 \coloneqq \gvar{\lambda\ t}])\\
  &\globnaive(t\ u)&&=\globnaive(t)\ \globnaive(u)\\
  &\globnaive(\gvar{t})&&=\gvar{t}
\end{alignat*}
\end{definition}
Due to the pre-condition on $\globnaive(t)$ that $t$ must be
closed, there is no need for a case for de Bruijn indices in the equations
above (a bare de Bruijn index is not closed). The pre-condition is maintained in
the recursion due to a substitution in case a $\lambda$ is encountered.

\begin{example}
  \label{ex:globalize-naive}
  The globalization of the term $\lambda\ \lambda\
  \debruijn 0\ \debruijn 1$ proceeds as follows:
  \[\globnaive(\lambda\ \lambda \ \debruijn 0\ \debruijn 1) =
  \lambda\ \globnaive(\lambda \ \debruijn 0\ g(\lambda\ \lambda \ \debruijn 0\ \debruijn 1)) =
  \lambda\ \lambda\ g(\lambda \ \debruijn 0\ \gvar{\lambda\ \lambda \ \debruijn 0\ \debruijn 1)}\ \gvar{\lambda\ \lambda \ \debruijn 0\ \debruijn 1}\]
\end{example}

In order
to understand how this algorithm works, we will first state the final theorem
that relates the algorithm to context-sensitive $\alpha$-equivalence.

\begin{theorem}[correctness of $\globnaive$]
  \label{thm:globalize-naive-correct}
  For $\lambda$-term nodes $\tn t_1[[p_1]]$ and $\tn t_2[[p_2]]$ we have $\tn
  t_1[[p_1]]\calphaeq \tn t_2[[p_2]]$ if and only if
  $\globnaive(t_1)[p_1] = \globnaive(t_2)[p_2]$.
\end{theorem}

The full proof of this theorem is postponed until Section~\ref{sec:proof-sketch}. Here,
we present a intuitive argument for why the algorithm works. The crux of the
algorithm lies in the property that closed $\lambda$-terms encoded with de Bruijn
indices are $\alpha$-equivalent if and only if they are syntactically equal.

\begin{lemma}[correctness of closed terms]
  For closed terms $t_1$, $t_2$ we have
  $\tn t_1[[\rootpos]] \calphaeq \tn t_2[[\rootpos]]$ iff $t_1 =
  t_2$.
\end{lemma}

See Lemma~\ref{lem:bisim-closed-equal} for a proof. Because the input of
$\globnaive$ is always closed, this lemma guarantees that the input term can
already be correctly compared. When a binder is encountered, we simply embed
this known-correct structure into a global variable and substitute it for any de
Bruijn index that references the binder. After the substitution, the subterm of
the binder is again closed and correct with respect to $\alpha$-equivalence. By
processing the entire term, all de Bruijn indices are replaced with a
global variable. At this point, every subterm is closed and can therefore be
compared syntactically with other (properly globalized) terms in order to
determine equality.

\subsection{Efficient Globalization}
\label{sec:alg-efficient}
The speed of $\globnaive(t)$ is dominated by the substitution we must perform
when we encounter a binder. Performing a substitution takes a linear amount of
time for a given term. Furthermore, a term of size $n$ may contain up to $O(n)$
binders. Therefore, in the worst case, $\globnaive(t)$ takes quadratic time.

\begin{example}
  \label{ex:naive-pathological}
  Consider the following pathological term of size $3n - 1$, containing $n$
  binders.
  \[\lambda x_1.\ \lambda x_2.\ \lambda x_3.\ \ldots\ \lambda x_n.\ x_n\ \ldots\ x_3\
    x_2\ x_1.\]
  The algorithm performs a substitution every time it encounters
  one of the $n$ binders. Further, each substitution must traverse a term whose
  size is at least $n$, resulting in at least $n^2$ steps.
\end{example}

To speed this
up, we would like to perform substitutions more lazily. If we accumulate
substitutions in a list $\sigma$, we can delay performing the substitution
$t\sigma$ until it is absolutely necessary.

How do we determine when $\sigma$ needs to be substituted? In the naive
algorithm, we rely on the property that the input term is always closed. Due to
lazy substitutions, we can no longer guarantee this. However, if a term is known
to not be $\alpha$-equivalent to any other term we might be interested in
comparing it to, even without the globalization substitutions, we can postpone
the substitution. Speeding up the algorithm relies on finding sufficiently many
subterms where we can skip substitution. Indeed, there are numerous simple
summaries that can be used to determine when a term is ``unique enough'' among a
set of other terms to skip the substitution step.

\begin{definition}[term summary]
  \label{def:term-summary}
  Let $|\cdot|$ be a function on $g$-terms to an arbitrary co-domain such that
  $\tn t_1[[p_1]] \forkeq \tn t_2[[p_2]]$ implies $|t_1[p_1]| = |t_2[p_2]|$.
\end{definition}

We will use term summaries in the contrapositive. That is, if the summary
$|t[p]|$ of a subterm is unique among the summaries of any other relevant subterm,
then it is not $\alpha$-equivalent to any of these subterms. We use the notation
$|\cdot|$ for term summaries because a rather useful example of a summary is the
size of the term: Any two $\alpha$-equivalent subterms are guaranteed
to have the same size. A stronger example of a term summary is the set of paths
$\valid{\cdot}$ of a term. Conversely, a rather weak example is the constant
function that maps every term to the same object. We need a summary that is
cheap to compute and compare, while distinguishing as many terms as possible.
The constant function is cheap but clearly distinguishes nothing. On the other
hand, $\valid{\cdot}$ distinguishes many terms but is expensive to compute and
compare. The size of a term is a good middle ground. It is cheap to compute,
cache, and compare while still distinguishing many terms.

\begin{lemma}
  The constant function, $\valid{\cdot}$, and the size of a term are valid term
  summaries.
\end{lemma}
\begin{proof}
  Straightforward from the definition of fork equivalence and Definition~\ref{def:lift}.
\end{proof}

We will use the term summary to find unique and non-unique terms. However, we
have not yet specified the background set to which a term should be compared for
uniqueness. Initially, one might think that we need to compare against the
entire infinite universe of potential $\lambda$-terms. Fortunately, that is not
the case. We can limit the set among which we need to compare to the
\textit{strongly connected component} of a term node.

\begin{definition}[strongly connected component]
  For a closed $g$-term $t$, define the \d{strongly connected component}
  $\SCC(t)$ as the set of all positions $p \in \valid{t}$ where for every
  nonempty prefix $p'$ of $p$, $t[p']$ is not closed.

\end{definition}

\begin{wrapfigure}[10]{r}{0.25\textwidth}
  \centerline{\includegraphics{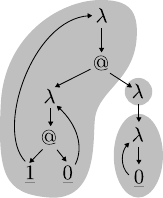}}
\end{wrapfigure}
We borrow the name ``strongly connected component'' from graph theory. In
particular, given a closed term $t$, the set \[\{ \tn t[[q]] \mid q \in \SCC(t)\}\]
forms the strongly connected sub-graph of nodes rooted in $t$.
\begin{example}
  The term $t = \lambda x. (\lambda y.\ x\ y)\ (\lambda y.\lambda z. z)$ is
  closed, and therefore we can ask what its strongly connected component is.
  Also note that $t$ contains two other closed subterms, for which we also
  have a strongly connected component. As such, there are three SCCs associated
  to $t$ and its subterms, as visualized on the right. Strongly connected
  components are disjoint and form a tree. SCCs may be singletons if the
  children of the root are closed.
\end{example}

Notice that because SCCs form a tree of closed terms, they can be processed
independently. If we know a procedure to globalize a single SCC, then we can
invoke this procedure recursively, either starting from the top-most SCC and
working our way down or the other way around. As such, we have reduced the
problem of globalization to individual strongly connected components. (This is a
common reduction for bisimulation algorithms.)

To efficiently globalize a SCC, we will calculate the set of subterms in the SCC
whose summary has one or more duplicate in the SCC. We will only need to
perform globalization substitutions in duplicate subterms because syntactic equality is
already strong enough to properly distinguish a non-duplicate term from all
other terms in the SCC.

\begin{definition}[duplicate SCC subterms]\label{def:duplicates}
  Let $\text{duplicates}(t)$ denote
  all strict subterms in the strongly connected component of $t$ whose term summary is not
  unique within the SCC.
  \[\text{duplicates}(t) = \{ t[p] \mid p,q \in \SCC(t) \wedge p,q \neq \rootpos
    \wedge p\neq q \wedge
    |t[p]| = |t[q]| \}\]
\end{definition}

One might object that this only guarantees globalization to give correct results
when comparing two subterms within the same SCC. What about two non-equivalent
subterms that do not share the same SCC? How do we guarantee that such terms are
not syntactically equal? For this, notice that two subterms can only be
$\alpha$-equivalent if all terms in their strongly connected component are also
pairwise $\alpha$-equivalent. In particular, the root of their respective SCCs
must be $\alpha$-equivalent. Because the root of a SCC is closed, it may be
safely syntactically compared. As such, where the naive algorithm would
substitute a $g$-var $\gvar{t}$, we can safely amend this to $\gvar{r\ t}$,
where $r$ is the root of the strongly connected component of $t$. This will
prevent us from inappropriately declaring two terms with different SCCs to be
$\alpha$-equivalent. We can now state an efficient globalization procedure.

\begin{definition}[efficient globalization]
  \label{def:eff-globalization}
  Recursively define $\glob$, $\globscc$ and
  $\globstep$:
  \begin{alignat*}{2}
    &\glob(r)&&=\globscc(r, [], r)\\[6pt]
    &\globscc(r, \sigma, \lambda\ t)&&=\lambda\ \globstep(r, (\gvar{r\ t} : \sigma), t) \\
    &\globscc(r, \sigma, t\ u)&&=\globstep(r, \sigma, t)\ \globstep(r, \sigma, u) \\
    &\globscc(r, \sigma, \gvar{t})&&=\gvar{t} \\
    &\globscc(r, \sigma, \debruijn i)&&=\debruijn i \sigma \\[6pt]
    &\globstep(r, \sigma, t)&&=
                               \begin{cases}
                                 \glob(t\sigma) & \text{if $t$ is closed or } t \in
                                                  \text{duplicates}(r)\\
                                 \globscc(r,\sigma,t) & \text{otherwise}
                               \end{cases}
  \end{alignat*}
  For $\glob(r)$, we maintain the precondition that $r$ is closed, while for
  $\globscc(r, \sigma, t)$ and $\globstep(r, \sigma, t)$ we maintain the
  precondition that $t\sigma$ is closed. Furthermore, it holds that there exists
  a position $p \in \SCC(r)$ such that $r[p] = t$. That is, $r$ is the
  root of the SCC in which $t$ resides. Finally, for $\globscc(r, \sigma, t)$, we expect that
  $t \not \in \text{duplicates}(r)$.
\end{definition}

Notice how the algorithm is defined mutually recursively between $\glob$,
$\globscc$ and
$\globstep$. There are two possible recursive paths, either back and forth
between $\globscc$ and $\globstep$, or with an detour through $\glob$.
Every time the algorithm calls $\glob(t\sigma)$, it crosses from one SCC to another. This
either happens because $t$ was already closed (and hence the start of a new SCC), or a new
SCC was created by performing the substitution $t\sigma$ because a duplicate was
found. The subsitution closes the term, creating a new SCC.

Similar to $\globnaive(t)$, we now claim that $\glob(t)$ behaves correctly with respect to
context-sensitive $\alpha$-equivalence.

\begin{theorem}[correctness of $\glob$]
  \label{thm:globalize-correct}
  For $\lambda$-term nodes $\tn t_1[[p_1]]$ and $\tn t_2[[p_2]]$ we have $\tn
  t_1[[p_1]]\calphaeq \tn t_2[[p_2]]$ if and only if
  $\glob(t_1)[p_1] = \glob(t_2)[p_2]$.
\end{theorem}

We again postpone the full proof of this theorem to Section~\ref{sec:proof-sketch}.
Although the explanations around the algorithm in this section
should provide a solid intuition, an airtight correctness proof requires more extended
reasoning. To further build intuition about this more elaborate algorithm, the
following observation shows that it is (nearly) a generalization of the naive
algorithm.

\begin{observation}
  When one instantiates the term summary $|\cdot|$ with a constant function, the
  $\glob(t)$ function reduces to a function that is very similar to
  $\globnaive(t)$. The only difference is the precise substitution being
  performed when a binder is encountered:
  \begin{align*}
    \globnaive(\lambda\ t) &= \lambda\ \globnaive(t[\debruijn 0 \coloneq \gvar{\lambda\ t}])\\
    \glob(\lambda\ t) &= \lambda\ \glob(t[\debruijn 0 \coloneq \gvar{t\ t}])
  \end{align*}
  Both substitutions lead to correct results. In fact, it is sufficient to
  simply substitute the $g$-var $\gvar{t}$. The variations $\gvar{\lambda\ t}$
  and $\gvar{t\ t}$ do not change the distinguishing power of the $g$-var.
\end{observation}

Now, given that the algorithm is known to behave correctly, we must ask the
question whether we have actually gained a substantial speed improvement. In the
beginning of this section, we attributed the source of inefficiency for the
naive algorithm to excessive substitutions. Interestingly, assuming that we
instantiate $|\cdot|$ to be term-size, the worst-case scenario presented in
Example~\ref{ex:naive-pathological} has now become a best-case scenario. This
is because the tree-structure of the example is almost entirely linear. The only
subterms with equal size are the variables (which all have size 1). Therefore, a
non-trivial duplicate is never encountered and substitution is only trivially
triggered when reaching a variable. Now, the substitutions take $O(n)$ time
instead of $O(n^2)$.\footnote{If substitution lists $\sigma$ are implemented
  using arrays, lookup and push operations take $O(1)$ amortized time.}

Conversely, the best-case (non-trivial) scenario for $\globnaive$ has now become
the worst-case scenario. Such a scenario involves a $\lambda$-term that forms a
perfectly balanced tree, where (almost) all subterms have a direct sibling that
is equal in size. This would cause the efficient algorithm to trigger a
substitution on every step. However, because the tree is now balanced, most
substitutions are performed on a small subterm. The substitutions then take at
most $O(n\log n)$ time.

To formalize this worst-case bound, we will show that each node in the syntax tree
is visited at most $O(\log n)$ times by the substitution function.\footnote{We
  analyze the cost of other functions, such as $\text{duplicates}(\cdot)$ in
  Section~\ref{sec:complexity}.} This is facilitated by assuming that every
$g$-term $t$ is annotated with a counter that is increased whenever the
substitution function traverses it. The visit count is retrieved using
$\text{sv}(t)$ and reset to zero (for all subterms) using $\text{sr}(t)$. We can
then prove the following efficiency lemma.

\begin{lemma}
  \label{lem:subst-logn}
  Let $n = \text{sv}(\glob(\text{sr}(t))[p])$. Then $|t| \geq 2^n$.
\end{lemma}
\begin{proof}
  By induction on $n$. The base case is trivial. For $n > 0$, we must unfold the
  algorithm until we reach the point where the first substitution occurs.
  Indeed, assuming that $u\sigma$ does not traverse into $u$ when $u$ is closed,
  a simple helper lemma can show the existence of $r$, $\sigma$, $q$ and $s$
  such that
  \begin{mathpar}
    \glob(t)[p] = \globstep(r, \sigma, r[q])[s] = \glob(r[q]\sigma)[s] \\
    |t| \geq |r| \and
    r[q] \in \text{duplicates}(r) \and
    \text{sv}(\glob(\text{sr}(r[q]\sigma))[s]) = n - 1.
  \end{mathpar}
  From the induction hypothesis, we then know that
  $|r[q]| = |r[q]\sigma| \geq 2^{n-1}$. Furthermore, we know there exists $q'$
  different from $q$ such that $|r[q']| = |r[q]|$. It then follows easily that
  $|t| \geq |r| \geq 2^n$.
\end{proof}

\begin{observation}
  The function $\globstep(r, \sigma, t)$ may sometimes substitute $\sigma$ even
  in cases where $t$ is not a binder. This is unnecessary. We can amend the
  algorithm to only perform a substitution when a binder is encountered. This
  will speed up the algorithm, but not asymptotically so. This optimization does
  somewhat complicate the proof of correctness. We omit these details.
\end{observation}


\subsection{A Concrete Hashing Implementation}
\label{sec:alg-concrete}

Although the efficient algorithm from the previous section can be shown to be
correct, there are some practical and theoretical shortcomings:
\begin{itemize}
\item The algorithm is not concrete enough to fully analyze its asymptotic
  complexity. In particular, the function $\text{duplicates}(t)$ is too
  abstract.
\item The use of $g$-terms to compare subterms for $\alpha$-equivalence is
  unsatisfactory because equality checking on $g$-terms takes $O(n)$ time.
  Instead, we would like to calculate a hash that can be compared in $O(1)$ time
  (at the expense of potential collisions).
\item The globalization process transforms $\lambda$-terms into $g$-terms,
  destroying any de Bruijn indices. This makes it difficult to further use the
  term as a normal $\lambda$-term. (Even though it is technically possible to recover
  the original $\lambda$-term from a globalized term, this is a non-trivial
  operation). We would like a globalization function that assigns appropriate
  hashes to each node of an $\lambda$-term, without modifying the term itself.
\end{itemize}

Here we present a more concrete algorithm implemented in the OCaml programming
language. A complete, executable reference implementation is
available~\cite{lasse_blaauwbroek_2023_10421517}.

\subsubsection{Datastructures}
\label{sec:datastructures}
We start with the definition of terms. We will need several variants
of $\lambda$-terms. To easily define them in a common framework, we define them
using a \textit{term functor}:
\begin{minted}{ocaml}
type 'a termf = Lam of 'a | Var of int | App of 'a * 'a [@@deriving map, fold]
\end{minted}
Instead of defining a term through direct recursion, we rather ``tie the knot'' in
this term functor. This allows us to decorate a term with additional information
when we need it by ``adding it to the knot.'' As an example, the simplest
recursive knot we can tie represents a pure, ordinary $\lambda$-term with no
additional information:
\begin{minted}{ocaml}
type pure_term = pure_term termf
\end{minted}
By adding an extra constructor \texttt{GVar} to the knot, we can also define a
structure that is isomorphic to $g$-terms:
\begin{minted}{ocaml}
type gterm = Term of gterm termf | GVar of gterm
\end{minted}
For our algorithm, we must efficiently calculate quite a few properties of
terms, including whether they are closed, their size and a hash. Information
related to this must be stored in each node of a term. Instead of providing a
concrete implementation for this, we rather posit the existence of an abstract
type \texttt{term} that is assumed to store all the required information. A
concrete implementation of this type can be found in supplementary
material~\cite{lasse_blaauwbroek_2023_10421517}. It comes with functions
\texttt{lift} and \texttt{case} that allows us to convert it to and from the
term functor, so that we can pattern match on it.
\begin{minted}{ocaml}
type term    val lift : term termf -> term    val case : term -> term termf
\end{minted}
The function \texttt{case} is the left inverse of \texttt{lift}, that is
\texttt{case (lift t) = t}. We do not have \texttt{lift (case t) = t}, because
information stored in \texttt{t} may be thrown away by \texttt{case}.
To illustrate how \texttt{lift} and \texttt{case} are used, we will write the
functions \texttt{from\_pure} and \texttt{to\_pure} that convert a
\texttt{pure\_term} into a \texttt{term} and vice versa.
For this, we will need the \texttt{map\_termf} function that has been automatically
derived for the term functor along with a \texttt{fold\_termf} function. They have the
following signature.
\begin{minted}{ocaml}
val map_termf  : ('a -> 'b) -> 'a termf -> 'b termf
val fold_termf : ('a -> 'b -> 'a) -> 'a -> 'b termf -> 'a
\end{minted}
We can use \texttt{map\_termf}, \texttt{lift} and \texttt{case} to write the
following recursive conversion functions.
\begin{minted}{ocaml}
let rec from_pure (t : pure_term) : term = lift (map_termf from_pure t)
let rec to_pure   (t : term) : pure_term = map_termf to_pure (case t)
\end{minted}
The \texttt{from\_pure} takes an ordinary $\lambda$-term, and lifts it into a
\texttt{term} decorated with information about term size, closedness and more.
Calculating the required information for this decoration happens in \texttt{lift}.
The \texttt{to\_pure} function does the opposite, because \texttt{case} will
forget any decorations that may be stored in the term.

The most important decoration of \texttt{term} is the hash we will assign to
each node through globalization. We consider two possible datatypes for a
\texttt{hash}. We can use integers as a hash if we want a datatype that is fast
to compare, but with the risk of encountering a collisions. When a collision is
not acceptable, we can use \texttt{gterm} as a hash. Here, we keep the datatype
for \texttt{hash} abstract (but keeping in mind our two target implementations):
\begin{minted}{ocaml}
type hash    val lift_hash : hash termf -> hash   val hash_gvar : hash -> hash
\end{minted}
In case \texttt{hash} is instantiated to be a \texttt{gterm}, we implement
\texttt{lift\_hash} and \texttt{hash\_gvar} as follows.
\begin{minted}{ocaml}
let lift_hash h = Term h and hash_gvar h = GVar h
\end{minted}
We assume a hash can be retrieved from any \texttt{term}
via function \texttt{hash}, with the following contract.
\begin{minted}{ocaml}
val hash : term -> hash
hash (lift t) = lift_hash (map_termf hash) t
\end{minted}
This means that when we convert a \texttt{pure\_term} into a \texttt{term}, the
hash for that term corresponds to the Merkle-style hash of its syntactic
structure (including de Bruijn indices). The idea of the globalization algorithm is to adjust these hashes by
annotating de Bruijn indices with a corrected, globalized hash.
To this end, we stipulate an alternative function for building a variable term with
a custom hash.
\begin{minted}{ocaml}
val gvar : hash -> int -> term
case (gvar h i) = Var i    &&    hash (gvar h i) = hash_gvar h
\end{minted}
Finally, we assume that we can retrieve the size of a term, and check if a term
is \texttt{gclosed}. This function returns false if and only if the given term
contains any free variable which was not built with \texttt{gvar}.
\begin{minted}{ocaml}
val size : term -> int        val gclosed : term -> bool
\end{minted}
\begin{wrapfigure}{r}{0.45\textwidth}
  \centerline{\includegraphics{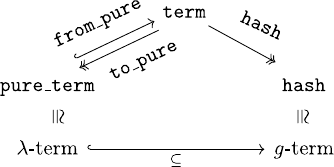}}
  \caption{OCaml datastructures and their mathematical counterparts as a commutative diagram.}
  \label{fig:datastructures}
\end{wrapfigure}
Figure~\ref{fig:datastructures} summarizes the datastructures we have built. A
\texttt{pure\_term} is isomorphic to a mathematical $\lambda$-term, and a
\texttt{hash} is isomorphic to a $g$-term (if the hash is instantiated as a
\texttt{gterm} and not an \texttt{int}). One can see a \texttt{term} as a pair
that contains a \texttt{pure\_term} and a \texttt{hash}.

\subsubsection{Calculating Duplicates}
We now turn our attention to the efficient calculation of $\text{duplicates}(r)$
from Definition~\ref{def:duplicates}. Note that
this set actually contains more terms than we need. In particular, for any $t
\in \text{duplicates}(r)$, we have no need for further sub-terms of $t$ to be
included in the set. This is because the algorithm is guaranteed
to transition from $\globscc$ to $\glob$ once it encounters $t$, which means
that a new SCC with different duplicates will become active. It is not difficult
to show that the algorithm behaves identically when we omit these irrelevant terms.

To efficiently calculate this reduced set of duplicate node terms, we require
the property that $|t| > |t[p]|$ for any $p \in \valid{t}$. This is satisfied by
instantiating the term summary with term size: The size of a subterm of $t$ is
smaller than the size of $t$ itself. We can now find duplicates by inserting
terms into a priority queue keyed to the size of the terms. We start with a
singleton queue that only contains the root of an SCC. Then, we retrieve all
terms whose key is equal to the largest key in the queue. Initially, this is
only the root of the SCC. If we have retrieved multiple terms, we know that they
are duplicates of each other. If we have retrieved only a single term, we know
that it cannot have a duplicate because all other terms in the queue are
smaller. We then insert the children of that unique term into the queue. We
iteratively retrieve and re-insert items into the queue until we have exhausted
all terms in the SCC. In OCaml code, this procedure is as follows.

\begin{minted}{ocaml}
val Heap.pop_multiple : Heap.t -> int * term list * Heap.t

let calc_duplicates (t : term) : IntSet.t =
  let step q t = if gclosed t then q else Heap.insert t q in
  let rec aux queue =
    match Heap.pop_multiple queue with
    | None -> IntSet.empty
    | Some (_,  [t], queue) -> aux (fold_termf step queue (case t))
    | Some (size, _, queue) -> IntSet.add size (aux queue)
  in aux (Heap.insert t Heap.empty)
\end{minted}

Note that unlike $\text{duplicates}(t)$, \texttt{calc\_duplicates(t)} does not
output a set of duplicate terms. Instead, it outputs a set of duplicate sizes.
To check if a term is duplicated in a SCC, one can simply check if the size of
that term is duplicated.

\subsubsection{Globalization}
\label{sec:fast-hashing}
Before we can define our \texttt{globalize} function, we must first define an
OCaml equivalent to substitutions. On the mathematical level, we substitute
$g$-vars for de Bruijn indices. The corresponding concept on the OCaml level is
to decorate a de Bruijn index with a hash. This is done through the function
\texttt{set\_hash}:
\begin{minted}{ocaml}
let rec set_hash (n : int) (h : hash) (t : term) : term =
  if gclosed t then t (* do not modify existing g-vars *)
  else match case t with
    | Lam t -> lift (Lam (set_hash (n+1) h t))
    | Var i -> if n = i then gvar h i else t
    | t     -> lift (map_termf (set_hash n h) t)
\end{minted}
A substitution $t[\debruijn i \coloneq \gvar{u}]$ can be seen as roughly
equivalent to \texttt{set\_hash i u t}. In addition to setting a single hash,
we must have the ability to set a sequence of hashes, similar to a substitution
$t\sigma$. For this, we have a datatype \texttt{hashes}, which is morally just a
list of hashes. However, a naive linked list would be too inefficient for
lookup. A more efficient implementation based on sets is out of scope of this
text. Instead, we specify \texttt{hashes} as an abstract datatype with the
following functions.

\begin{minted}{ocaml}
type hashes                  val push_hash  : hashes -> hash -> hashes
val empty_hashes : hashes    val set_hashes : hashes -> term -> term
\end{minted}
A simultaneous substitution $t\sigma$ can be seen as roughly equivalent to
\texttt{set\_hashes sigma t}.

We are now ready to write our globalization function in OCaml. The following is
essentially a direct transliteration of the equations from
Definition~\ref{def:eff-globalization} to OCaml.
\begin{minted}{ocaml}
let rec globalize (r : term) : term =
  let duplicates = calc_duplicates r in
  let rec globalize_scc (s : hashes) (t : term) =
    match case t with
    | Lam t' ->
      let s = push_hash s (hash (lift (App (r, t)))) in
      lift (Lam (globalize_step s t'))
    | Var _ -> set_hashes s t
    | t -> lift (map_termf (globalize_step s) t)
  and globalize_step s t =
    let t = if IntSet.mem (size t) duplicates then set_hashes s t else t in
    if gclosed t then globalize t else globalize_scc s t
  in globalize_scc empty_hashes r
\end{minted}
Following the correctness statement for the mathematical version of the
algorithm in Theorem~\ref{thm:globalize-correct}, we can state the correctness
of the OCaml version as follows. Note that this theorem relies on extending term
indexing $t[p]$ to the OCaml realm.
\begin{theorem}
  \label{thm:alg-correct}
  Let $\tn t_1[[p_1]]$, $\tn t_2[[p_2]]$ be two term nodes, and
  \texttt{t1}, \texttt{t2} the canonical embeddings of $t_1$, $t_2$ as
  an OCaml \texttt{term}.
  If $\tn t_1[[p_1]] \bisim \tn t_2[[p_2]]$, then
  \[
  \texttt{hash ((globalize t1)}[p_1]\texttt{)} =
  \texttt{hash ((globalize t2)}[p_2]\texttt{)}
  \]
  The reverse implication is true if \texttt{lift\_hash} and
  \texttt{hash\_gvar} are injective, and have disjoint images.
\end{theorem}

We state this theorem without further proof. However, it is straightforward to verify that
\[\texttt{hash ((globalize (from\_pure t))[p])} = \text{globalize}(t)[p]\]
if one instantiates the type \texttt{hash} with \texttt{gterm}. This provides a
clear link between the mathematical algorithm and the OCaml algorithm.

\subsubsection{Algorithmic Complexity}\label{sec:complexity}

We will now show that the algorithm presented in Section~\ref{sec:fast-hashing}
runs in $O(n\log n)$ time, where $n$ is the size of the term being globalized.
In Lemma~\ref{lem:subst-logn} we already showed that the \texttt{set\_hashes} function
touches each node at most $O(\log n)$ times. Furthermore, when a variable is
encountered for which a hash should be set, the lookup for the correct hash can
be done in $O(\log n)$ time. There are at most $n$ variables, and each variable
is assigned a hash exactly once. This demonstrates that the total cost of \texttt{set\_hashes}
remains within the budget.

To analyze the remaining functions, note that the traversal performed by the
mutually recursive functions \texttt{globalize}, \texttt{globalize\_scc} and
\texttt{globalize\_step} visits every node of a term exactly once. As such, it
suffices to verify that each invoked helper function spends no more than
$O(\log n)$ time per node. For most helper functions, like \texttt{gclosed},
\texttt{size} and \texttt{IntSet.mem} this is easy to verify.

The function \texttt{calc\_duplicates} is a bit more tricky. This function is
invoked once each time \texttt{globalize} is called with a fresh SCC. Its goal
is to calculate the set of nodes where we transition back from
\texttt{globalize\_scc} to \texttt{globalize}. As such, it touches exactly the
same set of nodes as the subsequent call to \texttt{globalize\_scc}. Therefore,
we can attribute the time taken for each node by \texttt{calc\_duplicates} to
this function call.
Processing a node entails inserting it into a queue in $O(\log n)$ time and
eventually retrieving it from the queue in $O(\log n)$ time. Therefore, we
stay within the available $O(\log n)$ time budget.

\section{Sketch of Correctness Proofs}
\label{sec:proof-sketch}
\begin{wrapfigure}[14]{r}{0.36\textwidth}
  \vspace{-2em}
  \includegraphics{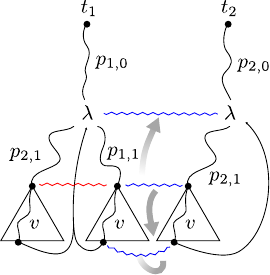}
  \caption{Schematic proof of Theorem~\ref{thm:bisim-to-fork}. Blue connections denote bisimilarity, red a single fork.}
  \label{fig:bisim-to-fork}
\end{wrapfigure}
The next three theorems give a sketch how to prove that bisimulation is equal to
fork-equivalence, and that the globalization algorithm is correct.
More fleshed out versions of these theorems can be found in
Appendix~\ref{sec:proofs}.
Each theorem is responsible for one of the implication arrows in
Figure~\ref{fig:proof-circular}.
We assume that $\tn t_1[[p_1]]$, $\tn t_2[[p_2]]$ are two $\lambda$-term nodes.

\begin{restatable}{theorem}{BisimToForkThm}
  \label{thm:bisim-to-fork}
  If $\tn t_1[[p_1]] \bisim \tn t_2[[p_2]]$, then
  $\tn t_1[[p_1]] \forkeq \tn t_2[[p_2]]$.
\end{restatable}
\begin{proofsketch}
  When two term nodes are bisimilar, we
  must find a sequence of places in their contexts where
  \textit{let}-abstractions can be introduced until the subjects become equal.
  Each \textit{let}-abstraction represents a single fork. Then, through
  transitivity, this sequence of single forks demonstrates fork equivalence.
  Finding this sequence of single forks proceeds by strong induction on the path
  $p_1$. That is, we assume the theorem holds for all strict prefixes of $p_1$.
  Then we make a split $p_1 = p_{1,0}\lamdown p_{1,1}$ such that $p_{1,1}$ is
  locally closed in $t_1[p_{1,0}]$ and there exists a free variable $v \in
  \free{t_1[p_1]}$ that references the binder at $t_1[p_{1,0}]$.\footnote{If no
    such split exists, $t_1[p_1]$ is closed, making the theorem trivial.} This
  situation is illustrated in Figure~\ref{fig:bisim-to-fork}. The bisimulation
  relation then guarantees that a similar split $p_2 = p_{2,0}\lamdown p_{2,1}$
  can be made such that $\tn t_1[[p_{1,0}]] \bisim \tn t_2[[p_{2,0}]]$. These
  two splits represent the bottom-most location where we introduce a
  \textit{let}-abstraction. All the remaining \textit{let}-abstractions that
  need to be introduced along the paths $p_{1,0}$ and $p_{2,0}$ are established
  through the induction hypothesis, which allows us to obtain $\tn
  t_1[[p_{1,0}]] \forkeq \tn t_2[[p_{2,0}]]$ and hence also $\tn
  t_1[[p_{1,0}\lamdown p_{2,1}]] \forkeq \tn t_2[[p_{2,0}\lamdown p_{2,1}]]$. We
  now only need to establish the single fork $\tn t_1[[p_{1,0}\lamdown p_{1,1}]]
  \sfork \tn t_1[[p_{1,0}\lamdown p_{2,1}]]$ (illustrated by a red connection in
  Figure~\ref{fig:bisim-to-fork}). This fork can be established using a technical
  lemma that relies on the fact that $p_{1,1}$ is locally closed in
  $t_1[p_{1,0}]$.\qedhere
  \par
\end{proofsketch}

\begin{restatable}{theorem}{ForkToAlgThm}
  \label{thm:fork-to-alg}
  If
  $\tn t_1[[p_1]] \forkeq \tn t_2[[p_2]]$, then
  $\glob(t_1)[p_1] = \glob(t_2)[p_2]$.
\end{restatable}
\begin{proofsketch}
  It suffices to show that two term nodes
  related through a single fork become equal after globalization. The full
  theorem then follows from transitivity of Leibniz equality. Hence, we need to
  show the conclusion assuming that the term nodes follow one of
  the two rules in Definition~\ref{def:single-fork}. For the second rule, where
  both subjects $t_1[p_1]$ and $t_2[p_2]$ are closed and equal, the conclusion
  follows readily because
  \[\glob(t_1)[p_1] = \glob(t_1[p_1]) = \glob(t_2[p_2]) = \glob(t_2)[p_2].\]
  This holds, because the globalization procedure only modifies the terms
  through substitutions, which cannot influence the closed subjects.

  Proving correctness for the first rule is more technical, but ultimately
  relies on the same strategy, where we move the indexing of positions $p_1$ and
  $p_2$ from outside $\glob$ to inside $\glob$.
\end{proofsketch}
\begin{restatable}{theorem}{AlgToBisimThm}
  \label{thm:alg-to-bisim}
  If
  $\glob(t_1)[p_1] = \glob(t_2)[p_2]$, then
  $\tn t_1[[p_1]] \bisim \tn t_2[[p_2]]$.
\end{restatable}
\begin{proofsketch}
  Here, we rely on a conservative extension of
  the bisimulation relation to $g$-terms such that we can show
  \[\tn t[[p]] \bisim \tn \glob(t)[p][[\rootpos]].\]
  That is, modulo this new bisimulation relation, the globalization algorithm
  does not modify the term at all. The final theorem then follows trivially. To
  make this work, we add an extra transition from nodes whose subject are a
  $g$-var to their corresponding binder. This transition does not use the
  context, but rather the knowledge about the context that has been stored
  inside the $g$-var by the globalization procedure.
  \end{proofsketch}

\section{Experimental Evaluation}
\label{sec:experimental-evaluation}

We evaluate and compare the runtime of our algorithm with three kinds of
synthetic $\lambda$-term. First, we uniformly sample closed terms of a fixed
size~\cite{DBLP:journals/corr/abs-2005-08856, DBLP:journals/jfp/GrygielL13}.
Second, we generate the unbalanced terms from
Example~\ref{ex:naive-pathological}. Third, are perfectly balanced terms such
that binders and applications are alternated. These latter two represent two
extreme cases an algorithm must handle.

The left plot of Figure~\ref{fig:evaluation} compares the naive algorithm of
Section~\ref{sec:naive-alg} to the efficient algorithm of
Section~\ref{sec:alg-efficient}. The trend shows that all terms can be
globalized in roughly $O(n\log n)$ time. The naive algorithm takes $O(n^2)$
time, with the exception of the best-case scenario of balanced terms.

The right plot of Figure~\ref{fig:evaluation} compares our algorithm with
Valmari's deterministic finite automaton minimization
algorithm~\cite{DBLP:journals/ipl/Valmari12} and the hashing algorithm of
Maziarz et al.~\cite{DBLP:conf/pldi/MaziarzELFJ21}. One should note that
these comparisons are not apples-to-apples, see Section~\ref{sec:versus-bisim}
and \ref{sec:versus}. The plot includes a version of our algorithm with and
without hash-consing. The hash-consing version should be compared to Valmari's
algorithm, as it can be used to assign equivalence classes to term nodes
without collisions.

We see that hash-consing $g$-terms imposes a significant performance overhead.
Nevertheless, it is still competitive with Valmari's algorithm. The performance
of our algorithm is close to Maziarz'. For linear terms, we can see the extra
$O(\log n)$ runtime factor emerge in Maziarz' algorithm.

\begin{figure}
\begin{tikzpicture}
  \begin{groupplot}[
    group style={
      group size=2 by 1,
      y descriptions at=edge left,
      horizontal sep=5pt,
    },
    width=0.555\textwidth,
    xmode=log,
    log basis x={2},
    ymajorgrids,
    xmajorgrids,
    xmin=16,
    xlabel={term size},
    ylabel={processing time (s)},
    legend pos = south east,
    legend cell align={left},
    ]
    \nextgroupplot[
    ymode=log,
    log basis y={10},
    ]
    \addlegendimage{white, mark=none}
    \addlegendentry{\underline{Algorithm}}
    \addlegendimage{black, dashed, mark=none}
    \addlegendentry{Naive}
    \addlegendimage{black, mark=none}
    \addlegendentry{Efficient}
    \addplot+[Dark2-A, mark = o, mark size=0.5pt] table [x=x, y=q2, col sep=tab] {data/linear_efficient.tsv};
    \label{plot:linear}
    \addplot+[Dark2-A, mark = o, mark size=0.5pt, dashed] table [x=x, y=q2, col sep=tab] {data/linear_naive.tsv};
    \addplot+[Dark2-B, mark = o, mark size=0.5pt] table [x=x, y=q2, col sep=tab] {data/random_efficient.tsv};
    \label{plot:random}
    \addplot+[Dark2-B, mark = o, mark size=0.5pt, dashed] table [x=x, y=q2, col sep=tab] {data/random_naive.tsv};
    \addplot+[Dark2-C, mark = o, mark size=0.5pt] table [x=x, y=q2, col sep=tab] {data/balanced_efficient.tsv};
    \label{plot:balanced}
    \addplot+[Dark2-C, mark = o, mark size=0.5pt, dashed] table [x=x, y=q2, col sep=tab] {data/balanced_naive.tsv};

    \addplot+[domain=1:4096, dotted, black, mark = .] {0.00000002 * x^2} node[above] {$n^2$};
    \addplot+[domain=1:4096, dotted, black, mark = .] {0.00000006 * x * ln(x)/ln(2)} node[above] {$n\log n$};
    \addplot+[domain=1:4096, dotted, black, mark = .] {0.00000003 * x} node [above] {$n$};

    \node [draw,fill=white] at (rel axis cs: 0.21,0.82) {\shortstack[l]{
        \quad\quad \underline{Term type} \\
        \ref{plot:balanced} Balanced \\
        \ref{plot:random} Random \\
        \ref{plot:linear} Linear}};

    \nextgroupplot[
    ymode=log,
    log basis y={10},
    legend pos = south east,
    legend cell align={left},
    ]
    \addlegendimage{white, mark=none}
    \addlegendentry{\underline{Term type}}
    \addlegendimage{black, dashed, mark=none}
    \addlegendentry{Balanced}
    \addlegendimage{black, mark=none}
    \addlegendentry{Linear}

    \addplot+[Dark2-A, mark = o, mark size=0.5pt] table [x=x, y=oursint, col sep=tab] {data/compare_linear.tsv};
    \label{plot:oursint}
    \addplot+[Dark2-B, mark = o, mark size=0.5pt] table [x=x, y=ourscons, col sep=tab] {data/compare_linear.tsv};
    \label{plot:ourscons}
    \addplot+[Dark2-C, mark = o, mark size=0.5pt] table [x=x, y=mariarz, col sep=tab] {data/compare_linear.tsv};
    \label{plot:mariarz}
    \addplot+[Dark2-D, mark = o, mark size=0.5pt] table [x=x, y=valmari, col sep=tab] {data/compare_linear.tsv};
    \label{plot:valmari}

    \addplot+[Dark2-A, mark = o, mark size=0.5pt, dashed] table [x=x, y=oursint, col sep=tab] {data/compare_balanced.tsv};
    \addplot+[Dark2-B, mark = o, mark size=0.5pt, dashed] table [x=x, y=ourscons, col sep=tab] {data/compare_balanced.tsv};
    \addplot+[Dark2-C, mark = o, mark size=0.5pt, dashed] table [x=x, y=mariarz, col sep=tab] {data/compare_balanced.tsv};
    \addplot+[Dark2-D, mark = o, mark size=0.5pt, dashed] table [x=x, y=valmari, col sep=tab] {data/compare_balanced.tsv};

    \node [draw,fill=white] at (rel axis cs: 0.22,0.78) {\shortstack[l]{
        \quad\quad \underline{Algorithm} \\
        \ref{plot:oursint} Ours Hash \\
        \ref{plot:ourscons} Ours Cons \\
        \ref{plot:mariarz} Maziarz\\
        \ref{plot:valmari} Valmari}};
  \end{groupplot}
\end{tikzpicture}
\caption{Performance of several algorithms on synthetic $\lambda$-terms.}
\label{fig:evaluation}
\end{figure}
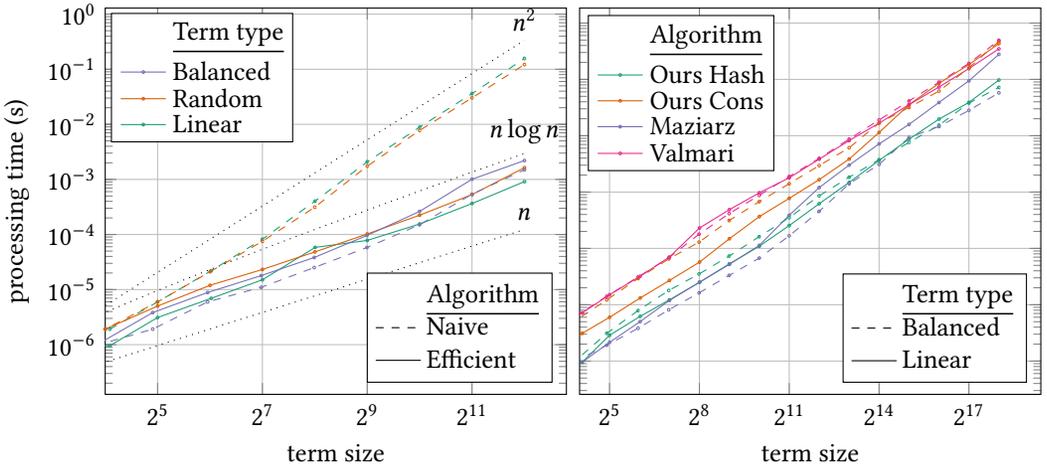

\section{Related and Future Work}

Our work should primarily be compared to previous work by Maziarz et
al.~\cite{DBLP:conf/pldi/MaziarzELFJ21} and bisimulation algorithms~\cite{DBLP:journals/siamcomp/PaigeT87, DBLP:journals/tcs/DovierPP04,
  hopcroft1971n, DBLP:journals/ipl/Valmari12}. This comparison can be found in
Section \ref{sec:versus-bisim} and \ref{sec:versus}. Here we give an overview of further related work
and future research.

\paragraph{Term Sharing Algorithms}
Term sharing is a common approach as a means of memory saving.
However, in most cases, these techniques do not take into account
$\alpha$-equality. In compilers, sharing
the structure of a languages AST is often achieved using
hash-consing~\cite{DBLP:conf/ml/FilliatreC06}. Hash-consing allows for
sub-structure sharing between terms, but shared terms are not guaranteed to be
``equal'' according to any reasonable equivalence relation. The FLINT
compiler~\cite{DBLP:conf/icfp/ShaoLM98} is an example where hash-consing is
employed aggressively to save space.

The literature is rather sparse with respect to term-sharing modulo
$\alpha$-equivalence. Condoluci, Accattoli and Coen present a decision procedure
to check $\alpha$-equivalence of two terms in which sub-terms may be shared in
linear time~\cite{DBLP:conf/ppdp/CondoluciAC19}. This is an important result
that may be used for efficient convertibility checking in dependently typed
proof assistants such as Coq,
LEAN and Agda~\cite{the_coq_development_team_2020_4021912,
  DBLP:conf/cade/MouraKADR15, DBLP:conf/tldi/Norell09} in
combination with efficient reduction algorithms that employ
sharing~\cite{DBLP:conf/fpca/BlellochG95, DBLP:journals/corr/AccattoliL16}.
However, their algorithm only allows pairwise comparisons of terms. It does not
show how to efficiently find all $\alpha$-equivalent subterms.

\paragraph{Hashing of Graphs}
We are not aware of existing work in labeled transition systems that calculates
a bsimimulation-respecting hash for each node. Such a hash would be useful in
the analysis of large-scale graphs, in which calculating the entire bisimulation
relation at once may not be feasible.
As such, an interesting open question is how far our algorithm can be
generalized for arbitrary graphs. The graphs induced by $\lambda$-calculus
are only a subset of the set of di-graphs. It is guaranteed that
during a traversal of a graph from the root, any binder is reached before a variable
that refers to that binder. Extending $\lambda$-calculus with mutually recursive
fixpoints eliminates this property. In such an extension, variables can no
longer be represented with de Bruijn indices, invalidating our algorithm. An
algorithm capable of hashing such terms is future work, as is extending the algorithm
to arbitrary (non-)deterministic transition systems.



\begin{acks}
  This work was partially supported by the Amazon Research Awards, EU ICT-48
  2020 project TAILOR no. 952215, and the European Regional Development Fund
  under the Czech project AI\&Reasoning with identifier
  CZ.02.1.01/0.0/0.0/15\_003/0000466. Lasse Blaauwbroek acknowledges travel
  support from the European Union’s Horizon 2020 research and innovation
  programme under grant agreement No 951847.
\end{acks}
\bibliographystyle{ACM-Reference-Format}
\bibliography{bibliography}


\begin{thebibliography}{28}


\ifx \showCODEN    \undefined \def \showCODEN     #1{\unskip}     \fi
\ifx \showDOI      \undefined \def \showDOI       #1{#1}\fi
\ifx \showISBNx    \undefined \def \showISBNx     #1{\unskip}     \fi
\ifx \showISBNxiii \undefined \def \showISBNxiii  #1{\unskip}     \fi
\ifx \showISSN     \undefined \def \showISSN      #1{\unskip}     \fi
\ifx \showLCCN     \undefined \def \showLCCN      #1{\unskip}     \fi
\ifx \shownote     \undefined \def \shownote      #1{#1}          \fi
\ifx \showarticletitle \undefined \def \showarticletitle #1{#1}   \fi
\ifx \showURL      \undefined \def \showURL       {\relax}        \fi
\providecommand\bibfield[2]{#2}
\providecommand\bibinfo[2]{#2}
\providecommand\natexlab[1]{#1}
\providecommand\showeprint[2][]{arXiv:#2}

\bibitem[Accattoli and Lago(2016)]%
        {DBLP:journals/corr/AccattoliL16}
\bibfield{author}{\bibinfo{person}{Beniamino Accattoli} {and}
  \bibinfo{person}{Ugo~Dal Lago}.} \bibinfo{year}{2016}\natexlab{}.
\newblock \showarticletitle{(Leftmost-Outermost) Beta Reduction is Invariant,
  Indeed}.
\newblock \bibinfo{journal}{\emph{Log. Methods Comput. Sci.}}
  \bibinfo{volume}{12}, \bibinfo{number}{1} (\bibinfo{year}{2016}).
\newblock
\urldef\tempurl%
\url{https://doi.org/10.2168/LMCS-12(1:4)2016}
\showDOI{\tempurl}


\bibitem[Baier and Katoen(2008)]%
        {DBLP:books/daglib/0020348}
\bibfield{author}{\bibinfo{person}{Christel Baier} {and}
  \bibinfo{person}{Joost{-}Pieter Katoen}.} \bibinfo{year}{2008}\natexlab{}.
\newblock \bibinfo{booktitle}{\emph{Principles of model checking}}.
\newblock \bibinfo{publisher}{{MIT} Press}.
\newblock
\showISBNx{978-0-262-02649-9}


\bibitem[Barendregt(1985)]%
        {BarendregtLambdaCalc}
\bibfield{author}{\bibinfo{person}{Hendrik~Pieter Barendregt}.}
  \bibinfo{year}{1985}\natexlab{}.
\newblock \bibinfo{booktitle}{\emph{The lambda calculus - its syntax and
  semantics}}. \bibinfo{series}{Studies in logic and the foundations of
  mathematics}, Vol.~\bibinfo{volume}{103}.
\newblock \bibinfo{publisher}{Elsevier, North-Holland}.
\newblock


\bibitem[Bendkowski(2020)]%
        {DBLP:journals/corr/abs-2005-08856}
\bibfield{author}{\bibinfo{person}{Maciej Bendkowski}.}
  \bibinfo{year}{2020}\natexlab{}.
\newblock \showarticletitle{How to generate random lambda terms?}
\newblock \bibinfo{journal}{\emph{CoRR}}  \bibinfo{volume}{abs/2005.08856}
  (\bibinfo{year}{2020}).
\newblock
\showeprint[arXiv]{2005.08856}
\urldef\tempurl%
\url{https://arxiv.org/abs/2005.08856}
\showURL{%
\tempurl}


\bibitem[Berghofer and Urban(2006)]%
        {DBLP:journals/entcs/BerghoferU07}
\bibfield{author}{\bibinfo{person}{Stefan Berghofer} {and}
  \bibinfo{person}{Christian Urban}.} \bibinfo{year}{2006}\natexlab{}.
\newblock \showarticletitle{A Head-to-Head Comparison of de Bruijn Indices and
  Names}. In \bibinfo{booktitle}{\emph{Proceedings of the First International
  Workshop on Logical Frameworks and Meta-Languages: Theory and Practice,
  LFMTP@FLoC 2006, Seattle, WA, USA, August 16, 2006}}
  \emph{(\bibinfo{series}{Electronic Notes in Theoretical Computer Science},
  Vol.~\bibinfo{volume}{174})}, \bibfield{editor}{\bibinfo{person}{Alberto
  Momigliano} {and} \bibinfo{person}{Brigitte Pientka}} (Eds.).
  \bibinfo{publisher}{Elsevier}, \bibinfo{pages}{53--67}.
\newblock
\urldef\tempurl%
\url{https://doi.org/10.1016/j.entcs.2007.01.018}
\showDOI{\tempurl}


\bibitem[Blaauwbroek(2023)]%
        {lasse_blaauwbroek_2023_10421517}
\bibfield{author}{\bibinfo{person}{Lasse Blaauwbroek}.}
  \bibinfo{year}{2023}\natexlab{}.
\newblock \bibinfo{booktitle}{\emph{{Reference Implementation for Hashing
  Modulo Context-Sensitive Alpha-Equivalence}}}.
\newblock
\urldef\tempurl%
\url{https://doi.org/10.5281/zenodo.11097757}
\showDOI{\tempurl}


\bibitem[Blaauwbroek(2024)]%
        {web-paper}
\bibfield{author}{\bibinfo{person}{Lasse Blaauwbroek}.}
  \bibinfo{year}{2024}\natexlab{}.
\newblock \showarticletitle{The {T}actician's Web of Large-Scale Formal
  Knowledge}.
\newblock \bibinfo{journal}{\emph{arXiv preprint}} (\bibinfo{date}{Jan.}
  \bibinfo{year}{2024}).
\newblock
\urldef\tempurl%
\url{https://doi.org/10.48550/arXiv.2401.02950}
\showDOI{\tempurl}
\showeprint[arxiv]{2401.02950}~[cs.LO]


\bibitem[Blelloch and Greiner(1995)]%
        {DBLP:conf/fpca/BlellochG95}
\bibfield{author}{\bibinfo{person}{Guy~E. Blelloch} {and} \bibinfo{person}{John
  Greiner}.} \bibinfo{year}{1995}\natexlab{}.
\newblock \showarticletitle{Parallelism in Sequential Functional Languages}. In
  \bibinfo{booktitle}{\emph{Proceedings of the seventh international conference
  on Functional programming languages and computer architecture, {FPCA} 1995,
  La Jolla, California, USA, June 25-28, 1995}},
  \bibfield{editor}{\bibinfo{person}{John Williams}} (Ed.).
  \bibinfo{publisher}{{ACM}}, \bibinfo{pages}{226--237}.
\newblock
\urldef\tempurl%
\url{https://doi.org/10.1145/224164.224210}
\showDOI{\tempurl}


\bibitem[Chargu{\'{e}}raud(2012)]%
        {DBLP:journals/jar/Chargueraud12}
\bibfield{author}{\bibinfo{person}{Arthur Chargu{\'{e}}raud}.}
  \bibinfo{year}{2012}\natexlab{}.
\newblock \showarticletitle{The Locally Nameless Representation}.
\newblock \bibinfo{journal}{\emph{J. Autom. Reason.}} \bibinfo{volume}{49},
  \bibinfo{number}{3} (\bibinfo{year}{2012}), \bibinfo{pages}{363--408}.
\newblock
\urldef\tempurl%
\url{https://doi.org/10.1007/S10817-011-9225-2}
\showDOI{\tempurl}


\bibitem[Chiusano et~al\mbox{.}({[n.\,d.]})]%
        {Chiusano_Bjarnason_Irani}
\bibfield{author}{\bibinfo{person}{Paul Chiusano}, \bibinfo{person}{Rúnar
  Bjarnason}, {and} \bibinfo{person}{Arya Irani}.}
  \bibinfo{year}{[n.\,d.]}\natexlab{}.
\newblock \bibinfo{booktitle}{\emph{Unison: A friendly, statically-typed,
  functional programming language from the future · UNISON programming
  language}}.
\newblock
\urldef\tempurl%
\url{https://www.unison-lang.org/}
\showURL{%
\tempurl}


\bibitem[Church(1941)]%
        {Church}
\bibfield{author}{\bibinfo{person}{Alonzo Church}.}
  \bibinfo{year}{1941}\natexlab{}.
\newblock \bibinfo{booktitle}{\emph{The Calculi of Lambda-Conversion}}.
\newblock \bibinfo{publisher}{Princeton: Princeton University Press}.
\newblock


\bibitem[Cocke(1970)]%
        {DBLP:conf/comop/Cocke70}
\bibfield{author}{\bibinfo{person}{John Cocke}.}
  \bibinfo{year}{1970}\natexlab{}.
\newblock \showarticletitle{Global common subexpression elimination}. In
  \bibinfo{booktitle}{\emph{Proceedings of a Symposium on Compiler
  Optimization, Urbana-Champaign, Illinois, USA, July 27-28, 1970}},
  \bibfield{editor}{\bibinfo{person}{Robert~S. Northcote}} (Ed.).
  \bibinfo{publisher}{{ACM}}, \bibinfo{pages}{20--24}.
\newblock
\urldef\tempurl%
\url{https://doi.org/10.1145/800028.808480}
\showDOI{\tempurl}


\bibitem[Condoluci et~al\mbox{.}(2019)]%
        {DBLP:conf/ppdp/CondoluciAC19}
\bibfield{author}{\bibinfo{person}{Andrea Condoluci},
  \bibinfo{person}{Beniamino Accattoli}, {and}
  \bibinfo{person}{Claudio~Sacerdoti Coen}.} \bibinfo{year}{2019}\natexlab{}.
\newblock \showarticletitle{Sharing Equality is Linear}. In
  \bibinfo{booktitle}{\emph{Proceedings of the 21st International Symposium on
  Principles and Practice of Programming Languages, {PPDP} 2019, Porto,
  Portugal, October 7-9, 2019}}, \bibfield{editor}{\bibinfo{person}{Ekaterina
  Komendantskaya}} (Ed.). \bibinfo{publisher}{{ACM}},
  \bibinfo{pages}{9:1--9:14}.
\newblock
\urldef\tempurl%
\url{https://doi.org/10.1145/3354166.3354174}
\showDOI{\tempurl}


\bibitem[de~Moura et~al\mbox{.}(2015)]%
        {DBLP:conf/cade/MouraKADR15}
\bibfield{author}{\bibinfo{person}{Leonardo~Mendon{\c{c}}a de Moura},
  \bibinfo{person}{Soonho Kong}, \bibinfo{person}{Jeremy Avigad},
  \bibinfo{person}{Floris van Doorn}, {and} \bibinfo{person}{Jakob von
  Raumer}.} \bibinfo{year}{2015}\natexlab{}.
\newblock \showarticletitle{The Lean Theorem Prover (System Description)}. In
  \bibinfo{booktitle}{\emph{Automated Deduction - {CADE-25} - 25th
  International Conference on Automated Deduction, Berlin, Germany, August 1-7,
  2015, Proceedings}} \emph{(\bibinfo{series}{Lecture Notes in Computer
  Science}, Vol.~\bibinfo{volume}{9195})},
  \bibfield{editor}{\bibinfo{person}{Amy~P. Felty} {and} \bibinfo{person}{Aart
  Middeldorp}} (Eds.). \bibinfo{publisher}{Springer},
  \bibinfo{pages}{378--388}.
\newblock
\urldef\tempurl%
\url{https://doi.org/10.1007/978-3-319-21401-6\_26}
\showDOI{\tempurl}


\bibitem[Dovier et~al\mbox{.}(2004)]%
        {DBLP:journals/tcs/DovierPP04}
\bibfield{author}{\bibinfo{person}{Agostino Dovier}, \bibinfo{person}{Carla
  Piazza}, {and} \bibinfo{person}{Alberto Policriti}.}
  \bibinfo{year}{2004}\natexlab{}.
\newblock \showarticletitle{An efficient algorithm for computing bisimulation
  equivalence}.
\newblock \bibinfo{journal}{\emph{Theor. Comput. Sci.}} \bibinfo{volume}{311},
  \bibinfo{number}{1-3} (\bibinfo{year}{2004}), \bibinfo{pages}{221--256}.
\newblock
\urldef\tempurl%
\url{https://doi.org/10.1016/S0304-3975(03)00361-X}
\showDOI{\tempurl}


\bibitem[Filli{\^{a}}tre and Conchon(2006)]%
        {DBLP:conf/ml/FilliatreC06}
\bibfield{author}{\bibinfo{person}{Jean{-}Christophe Filli{\^{a}}tre} {and}
  \bibinfo{person}{Sylvain Conchon}.} \bibinfo{year}{2006}\natexlab{}.
\newblock \showarticletitle{Type-safe modular hash-consing}. In
  \bibinfo{booktitle}{\emph{Proceedings of the {ACM} Workshop on ML, 2006,
  Portland, Oregon, USA, September 16, 2006}},
  \bibfield{editor}{\bibinfo{person}{Andrew Kennedy} {and}
  \bibinfo{person}{Fran{\c{c}}ois Pottier}} (Eds.). \bibinfo{publisher}{{ACM}},
  \bibinfo{pages}{12--19}.
\newblock
\urldef\tempurl%
\url{https://doi.org/10.1145/1159876.1159880}
\showDOI{\tempurl}


\bibitem[Grabmayer and Rochel(2014)]%
        {DBLP:conf/icfp/GrabmayerR14}
\bibfield{author}{\bibinfo{person}{Clemens Grabmayer} {and}
  \bibinfo{person}{Jan Rochel}.} \bibinfo{year}{2014}\natexlab{}.
\newblock \showarticletitle{Maximal sharing in the Lambda calculus with
  letrec}. In \bibinfo{booktitle}{\emph{Proceedings of the 19th {ACM} {SIGPLAN}
  international conference on Functional programming, Gothenburg, Sweden,
  September 1-3, 2014}}, \bibfield{editor}{\bibinfo{person}{Johan Jeuring}
  {and} \bibinfo{person}{Manuel M.~T. Chakravarty}} (Eds.).
  \bibinfo{publisher}{{ACM}}, \bibinfo{pages}{67--80}.
\newblock
\urldef\tempurl%
\url{https://doi.org/10.1145/2628136.2628148}
\showDOI{\tempurl}


\bibitem[Grygiel and Lescanne(2013)]%
        {DBLP:journals/jfp/GrygielL13}
\bibfield{author}{\bibinfo{person}{Katarzyna Grygiel} {and}
  \bibinfo{person}{Pierre Lescanne}.} \bibinfo{year}{2013}\natexlab{}.
\newblock \showarticletitle{Counting and generating lambda terms}.
\newblock \bibinfo{journal}{\emph{J. Funct. Program.}} \bibinfo{volume}{23},
  \bibinfo{number}{5} (\bibinfo{year}{2013}), \bibinfo{pages}{594--628}.
\newblock
\urldef\tempurl%
\url{https://doi.org/10.1017/S0956796813000178}
\showDOI{\tempurl}


\bibitem[Hopcroft(1971)]%
        {hopcroft1971n}
\bibfield{author}{\bibinfo{person}{John Hopcroft}.}
  \bibinfo{year}{1971}\natexlab{}.
\newblock \showarticletitle{An n log n algorithm for minimizing states in a
  finite automaton}.
\newblock In \bibinfo{booktitle}{\emph{Theory of machines and computations}}.
  \bibinfo{publisher}{Elsevier}, \bibinfo{pages}{189--196}.
\newblock


\bibitem[Maziarz et~al\mbox{.}(2021)]%
        {DBLP:conf/pldi/MaziarzELFJ21}
\bibfield{author}{\bibinfo{person}{Krzysztof Maziarz}, \bibinfo{person}{Tom
  Ellis}, \bibinfo{person}{Alan Lawrence}, \bibinfo{person}{Andrew~W.
  Fitzgibbon}, {and} \bibinfo{person}{Simon~Peyton Jones}.}
  \bibinfo{year}{2021}\natexlab{}.
\newblock \showarticletitle{Hashing modulo alpha-equivalence}. In
  \bibinfo{booktitle}{\emph{{PLDI} '21: 42nd {ACM} {SIGPLAN} International
  Conference on Programming Language Design and Implementation, Virtual Event,
  Canada, June 20-25, 2021}}, \bibfield{editor}{\bibinfo{person}{Stephen~N.
  Freund} {and} \bibinfo{person}{Eran Yahav}} (Eds.).
  \bibinfo{publisher}{{ACM}}, \bibinfo{pages}{960--973}.
\newblock
\urldef\tempurl%
\url{https://doi.org/10.1145/3453483.3454088}
\showDOI{\tempurl}


\bibitem[Norell(2009)]%
        {DBLP:conf/tldi/Norell09}
\bibfield{author}{\bibinfo{person}{Ulf Norell}.}
  \bibinfo{year}{2009}\natexlab{}.
\newblock \showarticletitle{Dependently typed programming in Agda}. In
  \bibinfo{booktitle}{\emph{Proceedings of TLDI'09: 2009 {ACM} {SIGPLAN}
  International Workshop on Types in Languages Design and Implementation,
  Savannah, GA, USA, January 24, 2009}},
  \bibfield{editor}{\bibinfo{person}{Andrew Kennedy} {and}
  \bibinfo{person}{Amal Ahmed}} (Eds.). \bibinfo{publisher}{{ACM}},
  \bibinfo{pages}{1--2}.
\newblock
\urldef\tempurl%
\url{https://doi.org/10.1145/1481861.1481862}
\showDOI{\tempurl}


\bibitem[Paige and Tarjan(1987)]%
        {DBLP:journals/siamcomp/PaigeT87}
\bibfield{author}{\bibinfo{person}{Robert Paige} {and}
  \bibinfo{person}{Robert~Endre Tarjan}.} \bibinfo{year}{1987}\natexlab{}.
\newblock \showarticletitle{Three Partition Refinement Algorithms}.
\newblock \bibinfo{journal}{\emph{{SIAM} J. Comput.}} \bibinfo{volume}{16},
  \bibinfo{number}{6} (\bibinfo{year}{1987}), \bibinfo{pages}{973--989}.
\newblock
\urldef\tempurl%
\url{https://doi.org/10.1137/0216062}
\showDOI{\tempurl}


\bibitem[Rute et~al\mbox{.}(2024)]%
        {graph2tac}
\bibfield{author}{\bibinfo{person}{Jason Rute}, \bibinfo{person}{Miroslav
  Olšák}, \bibinfo{person}{Lasse Blaauwbroek}, \bibinfo{person}{Fidel
  Ivan~Schaposnik Massolo}, \bibinfo{person}{Jelle Piepenbrock}, {and}
  \bibinfo{person}{Vasily Pestun}.} \bibinfo{year}{2024}\natexlab{}.
\newblock \showarticletitle{Graph2Tac: Learning Hierarchical Representations of
  Math Concepts in Theorem proving}.
\newblock \bibinfo{journal}{\emph{arXiv preprint}} (\bibinfo{date}{Jan.}
  \bibinfo{year}{2024}).
\newblock
\urldef\tempurl%
\url{https://doi.org/10.48550/arXiv.2401.02949}
\showDOI{\tempurl}
\showeprint[arxiv]{2401.02949}~[cs.LG]


\bibitem[Shao et~al\mbox{.}(1998)]%
        {DBLP:conf/icfp/ShaoLM98}
\bibfield{author}{\bibinfo{person}{Zhong Shao}, \bibinfo{person}{Christopher
  League}, {and} \bibinfo{person}{Stefan Monnier}.}
  \bibinfo{year}{1998}\natexlab{}.
\newblock \showarticletitle{Implementing Typed Intermediate Languages}. In
  \bibinfo{booktitle}{\emph{Proceedings of the third {ACM} {SIGPLAN}
  International Conference on Functional Programming {(ICFP} '98), Baltimore,
  Maryland, USA, September 27-29, 1998}},
  \bibfield{editor}{\bibinfo{person}{Matthias Felleisen}, \bibinfo{person}{Paul
  Hudak}, {and} \bibinfo{person}{Christian Queinnec}} (Eds.).
  \bibinfo{publisher}{{ACM}}, \bibinfo{pages}{313--323}.
\newblock
\urldef\tempurl%
\url{https://doi.org/10.1145/289423.289460}
\showDOI{\tempurl}


\bibitem[{The Coq Development Team}(2020)]%
        {the_coq_development_team_2020_4021912}
\bibfield{author}{\bibinfo{person}{{The Coq Development Team}}.}
  \bibinfo{year}{2020}\natexlab{}.
\newblock \bibinfo{booktitle}{\emph{The Coq Proof Assistant}}.
\newblock
\urldef\tempurl%
\url{https://doi.org/10.5281/zenodo.4021912}
\showDOI{\tempurl}


\bibitem[Thomsen and Henglein(2012)]%
        {DBLP:conf/iwsc/ThomsenH12}
\bibfield{author}{\bibinfo{person}{Mikkel~Jonsson Thomsen} {and}
  \bibinfo{person}{Fritz Henglein}.} \bibinfo{year}{2012}\natexlab{}.
\newblock \showarticletitle{Clone detection using rolling hashing, suffix trees
  and dagification: {A} case study}. In \bibinfo{booktitle}{\emph{Proceeding of
  the 6th International Workshop on Software Clones, {IWSC} 2012, Zurich,
  Switzerland, June 4, 2012}}, \bibfield{editor}{\bibinfo{person}{James~R.
  Cordy}, \bibinfo{person}{Katsuro Inoue}, \bibinfo{person}{Rainer Koschke},
  \bibinfo{person}{Jens Krinke}, {and} \bibinfo{person}{Chanchal~K. Roy}}
  (Eds.). \bibinfo{publisher}{{IEEE} Computer Society},
  \bibinfo{pages}{22--28}.
\newblock
\urldef\tempurl%
\url{https://doi.org/10.1109/IWSC.2012.6227862}
\showDOI{\tempurl}


\bibitem[Valmari(2012)]%
        {DBLP:journals/ipl/Valmari12}
\bibfield{author}{\bibinfo{person}{Antti Valmari}.}
  \bibinfo{year}{2012}\natexlab{}.
\newblock \showarticletitle{Fast brief practical {DFA} minimization}.
\newblock \bibinfo{journal}{\emph{Inf. Process. Lett.}} \bibinfo{volume}{112},
  \bibinfo{number}{6} (\bibinfo{year}{2012}), \bibinfo{pages}{213--217}.
\newblock
\urldef\tempurl%
\url{https://doi.org/10.1016/J.IPL.2011.12.004}
\showDOI{\tempurl}


\bibitem[Zhao et~al\mbox{.}(2015)]%
        {DBLP:conf/bwcca/ZhaoXFC15}
\bibfield{author}{\bibinfo{person}{Jingling Zhao}, \bibinfo{person}{Kunfeng
  Xia}, \bibinfo{person}{Yilun Fu}, {and} \bibinfo{person}{Baojiang Cui}.}
  \bibinfo{year}{2015}\natexlab{}.
\newblock \showarticletitle{An AST-based Code Plagiarism Detection Algorithm}.
  In \bibinfo{booktitle}{\emph{10th International Conference on Broadband and
  Wireless Computing, Communication and Applications, {BWCCA} 2015, Krakow,
  Poland, November 4-6, 2015}}, \bibfield{editor}{\bibinfo{person}{Leonard
  Barolli}, \bibinfo{person}{Fatos Xhafa}, \bibinfo{person}{Marek~R. Ogiela},
  {and} \bibinfo{person}{Lidia Ogiela}} (Eds.). \bibinfo{publisher}{{IEEE}
  Computer Society}, \bibinfo{pages}{178--182}.
\newblock
\urldef\tempurl%
\url{https://doi.org/10.1109/BWCCA.2015.52}
\showDOI{\tempurl}


\end{thebibliography}

\newpage
\appendix

\section{Proofs}
\label{sec:proofs}
The three sections that follow give a more detailed proof of each of the
theorems sketched in Section~\ref{sec:proof-sketch}.

\subsection{Bisimilarity implies Fork Equivalence}
\label{sec:bisim-to-fork}

We start with some preliminary observations and lemmas about the sets
$\valid{\cdot}$, $\vars{\cdot}$ and $\free{\cdot}$ and how they relate to term
indexing.

\begin{observation}
  \label{obs:valid-subterms}
  Various subsets of term paths can be derived from the paths of its subterms:
  \begin{alignat*}{5}
    &\valid{\lambda\ t} &&= \{ \lamdown p \mid p \in \valid{t}\} \cup \{\rootpos\}
    &&\valid{t\ u} &&= \{ \appleft p \mid p \in \valid{t}\}
    &&\cup \{ \appright p \mid p \in \valid{u}\} \cup \{\rootpos\} \\
    &\vars{\lambda\ t} &&= \{ \lamdown p \mid p \in \vars{t}\}
    &&\vars{t\ u} &&= \{ \appleft p \mid p \in \vars{t}\} &&\cup \{ \appright p \mid p \in \vars{u}\}\\
    &\free{\lambda\ t} &&= \{ \lamdown p \mid p \in \free{t} \wedge t[p] \neq \debruijn{\lamsize{p}}\}
    \quad&&\free{t\ u} &&= \{ \appleft p \mid p \in \free{t}\} &&\cup \{ \appright p \mid p \in \free{u}\}
  \end{alignat*}
\end{observation}

\begin{lemma}
  \label{lem:valid-empty-subterm}
  If $p \in \valid{t}$ then $\valid{t} = \valid{t[p]}$ implies $p = \rootpos$.
\end{lemma}
\begin{proof}
  By induction on $p$, using Observation~\ref{obs:valid-subterms}.
\end{proof}

The next three preliminary lemmas state that if two term nodes are bisimilar,
then the position sets $\valid{\cdot}$, $\vars{\cdot}$, $\free{\cdot}$ and
$\bound{\cdot}$ of their subjects must be equal. These lemmas are important,
because they show that bisimilar nodes have largely the same structure. One can
see the set $\valid{t}$ as the ``skeleton'' of $t$, where the contents of leafs
(variables) are ignored. If subjects do not have the same skeleton, there is no
hope of forming a bisimulation between them.

\begin{lemma}
  \label{lem:bisim-skeleton}
  If $\tn t_1[[p_1]] \bisim \tn t_2[[p_2]]$ then $\valid{t_1[p_1]} =
  \valid{t_2[p_2]}$.
\end{lemma}
\begin{proof}
  We have a bisimulation $R$ with $(\tn t_1[[p_1]], \tn
  t_2[[p_2]]) \in R$. Proceed by induction on $t_1[p_1]$.
  \begin{description}
  \item[Case {$t_1[p_1] = \debruijn i$}:]
    Relation $R$ mandates that there exists $j$ such that $t_2[p_2] =
    \debruijn j$. The conclusion is then trivially true.
  \item[Case {$t_1[p_1] = \lambda\ u$}:]
    We have $\tn t_1[[p_1]] \gstep \lamdown
    \tn t_1[[p_1\lamdown]]$. Furthermore, from $R$, we have that $\tn t_2[[p_2]]
    \gstep \lamdown \tn t_2[[p_2\lamdown]]$. The induction hypothesis then gives
    us $\valid{t[p_1\lamdown]} = \valid{t[p_2\lamdown]}$. Finally, we conclude
    $\valid{t[p_1]} = \valid{t[p_2]}$
    with the help of Observation~\ref{obs:valid-subterms}.
  \item[Case {$t_1[p_1] = u\ v$}:] Analogous to the previous case.\qedhere
  \end{description}
\end{proof}

Whereas the previous lemma shows that bisimilar subjects must have equal
skeletons, the next lemma shows that their variables must also be related. In
the introduction, we showed that the de Bruijn indices of bisimilar subjects are
not always equal. Nevertheless, positions that represent free variables in one
subject must also be free positions in the other subject. (The same fact holds
for bound variables because $\bound{\cdot}$ is the complement of $\free{\cdot}$.)

\begin{lemma}
  \label{lem:bisim-free}
  If $\tn t_1[[p_1]] \bisim \tn t_2[[p_2]]$ then $\free{t_1[p_1]} =
  \free{t_2[p_2]}$.
\end{lemma}
\begin{proof}
  We have a bisimulation $R$ with $(\tn t_1[[p_1]], \tn t_2[[p_2]]) \in R$.
  Proceed by induction on $t_1[p_1]$. Cases for variables and application
  proceed straightforward. The interesting case occurs when
  $t_1[p_1] = \lambda\ u$. From $R$ and the induction hypothesis, we obtain
  $\free{t_1[p_1\lamdown]} = \free{t_2[p_2\lamdown]}$.
  Observation~\ref{obs:valid-subterms} shows that it now suffices to prove
  \[\{ q \mid q \in \free{t_1[p_1\lamdown]} \wedge t_1[p_1\lamdown q] \neq \debruijn{\lamsize{q}}\}
    = \{ q \mid q \in \free{t_2[p_2\lamdown]} \wedge t_2[p_2\lamdown q] \neq
    \debruijn{\lamsize{q}\}}.\]
  In other words, it suffices to prove that if $q \in
  \free{t_1[p_1\lamdown]}$ and $t_1[p_1\lamdown q] = \debruijn{\lamsize{q}}$, then
  $t_2[p_2\lamdown q] = \debruijn{\lamsize{q}}$. Without loss of generality, we
  assume $t_2[p_2\lamdown q] < \debruijn{\lamsize{q}}$ to obtain a contradiction.
  We can then make a split $q = q_0 \lamdown q_1$ such that
  \[\tn t_2[[p_2]] \bisim \tn t_1[[p_1]] \bisim \tn t_2[[p_2\lamdown
    q_0\lamdown]].\]
  Then, from Lemma~\ref{lem:bisim-skeleton} we have $\valid{t_2[p_2]} =
  \valid{t_2[p_2\lamdown q_0 \lamdown]}$. Finally,
  Lemma~\ref{lem:valid-empty-subterm} gives a contradiction.
\end{proof}

Next is a lemma that shows that if two subjects are bisimilar, their
sub-structures must also be bisimilar. This is the analogous lemma to
Observation~\ref{obs:forkprop} for fork equivalence.

\begin{lemma}
  \label{lem:bisim-subpath}
  If $q \in \valid{t_1[p_1]}$ and $\tn t_1[[p_1]] \bisim \tn t_2[[p_2]]$ then $\tn t_1[[p_1q]] \bisim \tn t_2[[p_2q]]$.
\end{lemma}
\begin{proof}
  Straightforward by induction on $q$.
\end{proof}

As a final preliminary lemma, we note that if the subject of a term node is
closed, then its context is irrelevant. This is shown by establishing a
bisimulation relation between the node, and a modification such that the context
is thrown away.

\begin{lemma}
  \label{lem:bisim-closed}
  If $t$ and $t[p]$ are closed, then $\tn t[[p]] \bisim \tn t[p][[\rootpos]]$.
\end{lemma}
\begin{proof}
  Construct the relation
  \[ R = \{(\tn t[[pq]], \tn t[p][[q]]) \mid q \in \valid{t[p]}\}.\]
  Verifying that $R$ is a bisimulation relation is straightforward. The only
  case of note is when $t[pq]$ is a variable. We know that the binder
  corresponding to the variable is a subterm of $t[p]$, because that term is
  closed. Hence, we can verify that this binder is bisimilar to itself under $R$.
\end{proof}

We are now ready to prove the main technical ``workhorse'' lemma for this section. The
following lemma extracts the required information for a bisimulation relation in
order to establish a single fork. The conclusion of this lemma
corresponds closely to the required conditions in
Definition~\ref{def:single-fork} to build a single fork.
Note that the addition of the index into position $r$ is a technical requirement
to make the induction hypothesis sufficiently strong. When the lemma is used, we
always set $r = \rootpos$.

\begin{lemma}
  \label{lem:bisim-locally-closed-equal}
  If $\tn t[[pq_1]] \bisim \tn t[[pq_2]]$ and $q_1$ is locally closed in $t[p]$,
  then $\lift t[p]<q_1>[r] = \lift t[p]<q_2>[r]$.
\end{lemma}
\begin{proof}
  Note that from Lemma~\ref{lem:bisim-subpath} we have
  \begin{equation}
    \label{eq:blce-bisim-r}
    \tn t[[pq_1r]] \bisim \tn t[[pq_2r]].
  \end{equation}
  Proceed by induction on $t[pq_1r]$. In case $t[pq_1r]$ is a variable, we perform additional
  case analysis on whether $r$ is bound or free.
  \begin{description}
  \item[Case {$t[pq_1r] = \debruijn i$} such that $r \in \free{t[pq_1]}$:]
    Because $r$ is free, and $q_1$ is locally closed in $t[p]$, we know that
    $q_1r \in \free{t[p]}$. Therefore, there exists a split $p = p_1\lamdown
    p_2$ such that
    \[\tn t[[p_1\lamdown p_2q_1r]] \gstep\varup \tn t[[p_1]].\]
    The bisimulation relation from
    Equation~\ref{eq:blce-bisim-r} additionally mandates that $t[pq_2r] =
    \debruijn j$ for some $j$. Without loss of generality, assume $i \leq j$. We
    then know that there exists $s$ such that
    \begin{mathpar}
      \tn t[[pq_2r]] \gstep\varup \tn t[[p_1s]] \and
      \tn t[[p_1]] \bisim \tn t[[p_1s]].
    \end{mathpar}
    From Lemma~\ref{lem:bisim-skeleton} and Lemma~\ref{lem:valid-empty-subterm}
    we then have $s = \rootpos$, and as such
    \begin{mathpar}
      t[pq_1r] = \debruijn{\lamsize{p_2q_1r}} \and
      t[pq_2r] = \debruijn{\lamsize{p_2q_2r}}.
    \end{mathpar}
    We then conclude that
    \[\lift t[p]<q_1>[r] = \debruijn{\lamsize{p_2q_1r} - \lamsize{q_1}} =
      \debruijn{\lamsize{p_2q_2r} - \lamsize{q_2}} = \lift t[p]<q_2>[r].\]
  \item[Case {$t[pq_1r] = \debruijn i$} such that $r \in \bound{t[pq_1]}$:]\sloppy
    From the main bisimulation hypothesis and Lemma~\ref{lem:bisim-free} we have
    $\free{t[pq_1]} = \free{t[pq_2]}$. Since $\bound{\cdot}$ is the complement
    of $\free{\cdot}$, we know that $r \in \bound{t[pq_2]}$.
    The bisimulation relation from
    Equation~\ref{eq:blce-bisim-r} then mandates that there exist two splits
    $r = r_{1A}\lamdown r_{1B} =
    r_{2A}\lamdown r_{2B}$ such that
    \begin{mathpar}
      t[pq_1r] = \debruijn{\lamsize{r_{1B}}} \and
      t[pq_2r] = \debruijn{\lamsize{r_{2B}}} \and
      \tn t[[pq_1r_{1A}]] \bisim \tn t[[pq_2r_{2A}]].
    \end{mathpar}
    Moreover, from Lemma~\ref{lem:bisim-skeleton} we have $\valid{t[pq_1r_{1A}]} =
    \valid{t[pq_2r_{2A}]}$. Now, without loss of generality, assume $|r_{1A}| \geq
    |r_{2A}|$. We can then make an additional split $r_{1A} = r_{2A} s$, giving us
    \begin{mathpar}
      \valid{t[pq_1r_{2A} s]} = \valid{t[pq_2r_{2A}]}\and
      r = r_{2A} s \lamdown r_{1B} = r_{2A} \lamdown r_{2B} \and
      s \lamdown r_{1B} = \lamdown r_{2B}.
    \end{mathpar}
    Now, using the hypothesis $\tn t[[pq_1]] \bisim \tn t[[pq_2]]$ and
    Lemma~\ref{lem:bisim-subpath} we also have
    \[\valid{t[pq_1r_{2A}]} = \valid{t[pq_2r_{2A}]}.\]
    Putting this together, we get
    \[\valid{t[pq_1r_{2A}]} = \valid{t[pq_1r_{2A} s]}.\]
    Lemma~\ref{lem:valid-empty-subterm} then mandates $s = \rootpos$. This concludes
    the case, because it implies $r_{1B} = r_{2B}$ and hence
    \[\lift t[p]<q_1>[r] = t[pq_1r] =
      \debruijn{\lamsize{r_{1B}}} = \debruijn{\lamsize{r_{2B}}} = t[pq_2r] = \lift t[p]<q_2>[r].\]
  \item[Case {$t[p] = \lambda\ u$} and {$t[p] = u\ v$}:] These cases follow
    by straightforward application of the induction hypothesis.\qedhere
  \end{description}
\end{proof}

In the informal discussion of the algorithm, we have repeatedly referenced the
fact that two closed subjects without a context are $\alpha$-equivalent if and
only the subjects are equal. This fact is a corollary of the technical lemma above.

\begin{lemma}
  \label{lem:bisim-closed-equal}
  For closed terms $t_1$, $t_2$ we have
  $\tn t_1[[\rootpos]] \calphaeq \tn t_2[[\rootpos]]$ iff $t_1 =
  t_2$.
\end{lemma}
\begin{proof}
  The right-to-left implication follows directly from the fact the bisimulation
  relation is reflexive.
  For the left-to-right implication, we will use
  Lemma~\ref{lem:bisim-locally-closed-equal} instantiated with
  \begin{mathpar}
    t \coloneq t_1\ t_2 \and
    p\coloneq \rootpos \and
    q_1 \coloneq\ \appleft\and
    q_2 \coloneq\ \appright\and
    r \coloneq \rootpos
  \end{mathpar}
  The bisimilarity precondition is obtained with the help of
  Lemma~\ref{lem:bisim-closed}:
  \[\tn (t_1\ t_2)[[\appleft]] \bisim \tn t_1[[\rootpos]] \bisim \tn
    t_2[[\rootpos]] \bisim \tn (t_1\ t_2)[[\appright]].\]
  Lemma~\ref{lem:bisim-locally-closed-equal} then lets us conclude
  \[t_1 = (t_1\ t_2)[\appleft] = \lift (t_1\ t_2)<\appleft> = \lift (t_1\
    t_2)<\appright> = (t_1\ t_2)[\appright] = t_2.\qedhere\]
\end{proof}

Now, we have all the basic ingredients to prove Theorem~\ref{thm:bisim-to-fork}.
As shown in Figure~\ref{fig:double-fork} in the introduction, sometimes, we need
multiple different subforks to establish a fork equivalence. In the proof of
Theorem~\ref{thm:bisim-to-fork}, we use strong induction to decompose
a bisimilar pair $\tn t_1[[p_1]]\bisim \tn t_2[[p_2]]$ into a sequence of
single forks. In the base case, the information required to form the single fork
comes from Lemma~\ref{lem:bisim-closed-equal}. In the step case, the required
information is extracted using Lemma~\ref{lem:bisim-locally-closed-equal}.

\BisimToForkThm*
\begin{proof}
  Proceed by strong induction on $p_1$. That is, we suppose that the claim is
  true for any strict prefix of $p_1$ and other arguments are changed
  arbitrarily. Let us split $p_1 = p'_{1,0}p_{1,1}$ so that $p_{1,1}$ is locally
  closed in $t[p'_{1,0}]$, and $p'_{1,0}$ is as short as possible.
  The proof proceeds differently whether $p'_{1,0}$ is empty or not.
  \begin{description}
  \item[Case {$p'_{1,0} = \rootpos$}:]
    We know that $t_1[p_1]$ is closed. Furthermore, from
    Lemma~\ref{lem:bisim-free} we know that $t_2[p_2]$ is also closed.
    Therefore, using Lemma~\ref{lem:bisim-closed} we have
    \[\tn t_1[p_1][[\rootpos]] \bisim \tn t_1[[p_1]] \bisim \tn t_2[[p_2]]
      \bisim \tn t_2[p_2][[\rootpos]].\]
    Using Lemma~\ref{lem:bisim-closed-equal} we then obtain $t_1[p_1] =
    t_2[p_2]$. We can then directly establish $\tn t_1[[p_1]] \sfork \tn t_2[[p_2]]$
    using the second rule of Definition~\ref{def:single-fork}.
  \item[Case {$p'_{q,0} \neq \rootpos$}:]
    By definition of the split $p'_{1,0}p_{1,1}$, if we move the last symbol from $p'_{1,0}$ to the beginning of $p_{1,1}$,
    $p_{1,1}$ stops being locally closed in $t[p'_{1,0}]$. Therefore, this
    symbol is $\lamdown$. Let us then denote $p_{1,0}'$
    without it as $p_{1,0}$, so that $p = p_{1,0}\lamdown p_{1,1}$.
    Since $\lamdown p_{1,1}$ is not locally closed in $t[p_{1,0}]$, there is a $v\in\free{t_1[p_1]}$ such that
    $\tn t_1[[p_1v]] \gstep\varup \tn t_1[[p_{1,0}]]$. By Lemma~\ref{lem:bisim-free}, also $v\in\free{t_2[p_2]}$, and
    there is a split $p_2 = p_{2,0}\lamdown p_{2,1}$ such that
    \begin{mathpar}
      \tn t_2[[p_2v]] \gstep\varup \tn t_2[[p_{2,0}]] \and
      \tn t_1[[p_{1,0}]]\bisim \tn t_2[[p_{2,0}]].
    \end{mathpar}

    We use induction assumption on $\tn t_1[[p_{1,0}]]\bisim \tn t_2[[p_{2,0}]]$,
    and together with Observation~\ref{obs:forkprop} obtain
    \[
      \tn t_1[[p_{1,0}\lamdown p_{2,1}]]\forkeq \tn t_2[[p_{2,0}\lamdown p_{2,1}]].
    \]
    To finish the proof, we need to prove
    that $\tn t_1[[p_{1,0}\lamdown p_{1,1}]] \forkeq \tn t_1[[p_{1,0}\lamdown p_{2,1}]]$.
    We know that these two term nodes are bisimilar by
    \[
      \tn t_1[[p_{1,0}\lamdown p_{1,1}]] \bisim \tn t_2[[p_{2,0}\lamdown p_{2,1}]] \bisim \tn t_1[[p_{1,0}\lamdown p_{2,1}]].
    \]
    Lemma~\ref{lem:bisim-locally-closed-equal} then gives us
    \[\lift t_1[p_{1,0}]<\lamdown p_{1,1}> = \lift t_1[p_{1,0}]<\lamdown
      p_{2,1}>.\]
    The required fork can then be established by using the first rule of
    Definition~\ref{def:single-fork}.
    A schematic overview of this case can be found in
    Figure~\ref{fig:bisim-to-fork}.
    \qedhere
  \end{description}
\end{proof}

\subsection{Fork Equivalence implies Algorithm}
\label{sec:fork-to-alg}


Before we start analyzing the behavior of the globalization algorithm, we first
make some preliminary observations about the interaction between term indexing,
closed terms, substitutions, locally closed positions, and strongly connected
components.

\begin{observation}
  \label{obs:locally-closed-closed}
  Let $p \in \valid{t}$ be locally closed in $t$. If either $t$ is closed or $\lift
  t<p>$ is closed, then $\lift t<p> = t[p]$.
\end{observation}

\begin{observation}
  \label{obs:lift-subst}
  If $p$ is locally closed in $t$ then $p$ is locally closed in $t\sigma$.
  Further, $\lift t\sigma<p> = \lift
  t<p>\sigma$.
\end{observation}

\begin{observation}
  \label{obs:locally-closed-open}
  Let $t$ be a closed term such that $p \in \SCC(t)$ and $q$ is
  locally closed in $t[p]$. If $\lift t[p]<q>$ is open, then $pq \in \SCC(t)$.
\end{observation}

Now we can start analyzing the behavior of the globalization algorithm. We start
with a simple lemma that states that the context around a closed subject does
not influence the behavior of $\glob$ and its helper functions on that subject.
This lemma will later be used to show if there is a single fork between term
nodes formed through the second rule of Definition~\ref{def:single-fork}, then
the closed subjects of those nodes must be equal after globalization.

\begin{lemma}
  \label{lem:closed-to-alg}
  If $p \in \valid{t}$ and $t[p]$ is closed, then for all $r$ and $\sigma$,
  \[
    \glob(t)[p] = \globscc(r,\sigma,t)[p] = \globstep(r,\sigma,t)[p] = \glob(t[p]).
  \]
\end{lemma}
\begin{proof}
  We proceed by induction on $p$.
  If $p = \rootpos$ then $t$ is closed,
  and hence $t\sigma = t$. The conclusion follows directly from the definitions.
  Otherwise, when $p = xp_0$, we can perform a single unfolding of the
  definitions. The conclusion then follows directly from the induction hypothesis.
\end{proof}

The remaining lemmas are meant to analyze the behavior of $\glob$ for term nodes
with a single fork formed through the first rule of
Definition~\ref{def:single-fork}. Eventually, we will argue that a this rule
gives rise to a term whose SCC contains a duplicate term. The following
lemma demonstrates that this allows us to move an indexing operation that selects this
duplicate term from outside $\globstep$ to inside $\glob$. This is a similar
idea to the previous lemma, but now following a different path of the algorithm,
were we know that a non-trivial substitution will occur.

\begin{lemma}
  \label{lem:dup-globstep-to-glob}
  Let $p \in \valid{t}$, $t\sigma[p]$ be closed, and $t[p] \in \text{duplicates}(r)$. Then
  \[\globstep(r, \sigma, t)[p] = \glob(t\sigma[p]).\]
\end{lemma}
\begin{proof}
  By induction on $p$. When $p = \rootpos$, the equality holds trivially.
  Furthermore, if $t$ is closed or $t \in \text{duplicates}(r)$, we have
  \[\globstep(r, \sigma, t)[p] = \glob(t\sigma)[p].\]
  The problem then reduces to Lemma~\ref{lem:closed-to-alg}.
  Otherwise, the most interesting case is $p =\ \lamdown p_0$.
  From the induction hypothesis, we get
  \[\globstep(r, \sigma, t)[\lamdown p_0] =
    \globstep(r, (\gvar{r\ t} : \sigma), t[\lamdown])[p_0] =
    \glob(t[\lamdown](\gvar{r\ t} : \sigma)[p_0]).\]
  The proof is then completed by realizing that
  \[t[\lamdown](\gvar{r\ t} : \sigma)[p_0] =
    t\sigma[\lamdown p_0]\]
  This is true because $t\sigma[\lamdown p_0]$ is closed, and therefore the
  topmost $\lambda$ of $t$ is never referenced.
\end{proof}

The following lemma contains the core argument that allows us to conclude
correct behavior of $\glob$ on subjects of nodes with a single fork formed
through the first rule. Note how some of the assumptions of this lemma correspond
closely to the preconditions of this rule.

\begin{lemma}
  \label{lem:fork0-to-alg}
  Let $r$ be closed, $p \in \SCC(r)$ and $r[p]\sigma$ be closed.
  Let $q_1, q_2\in \valid{r[p]}$ be locally closed positions in $r[p]$ such that
  $\lift r[p]<q_1> = \lift r[p]<q_2>$. Then
  \[
  \globstep(r,\sigma,r[p])[q_1] = \globstep(r,\sigma,r[p])[q_2].
  \]
\end{lemma}
\begin{proof}
  First, consider the case where $\lift r[p]<q_1>$ is closed. From
  Observation~\ref{obs:locally-closed-closed} we then have
  \[
    r[p][q_1] = \lift r[p]<q_1> = \lift r[p]<q_2> = r[p][q_2].
  \]
  By Lemma~\ref{lem:closed-to-alg}, both sides of the desired equation now
  reduce to $\glob(r[pq_1])$, making it true by reflexivity.

  We can now assume that neither $\lift r[p]<q_1>$ nor $\lift r[p]<q_2>$ is
  closed. From Observation~\ref{obs:locally-closed-open} we then have $pq_1,
  pq_2 \in \SCC(r)$. Furthermore, the definition of a single fork
  (\ref{def:single-fork}) gives us $\tn r[p][[q_1]] \sfork \tn r[p][[q_1]]$, and
  hence, according to the definition of term summaries (\ref{def:term-summary}),
  $|r[pq_1]| = |r[pq_2]$|. These facts give us
  \[r[pq_1], r[pq_2] \in \text{duplicates}(r).\]
  Using Lemma~\ref{lem:dup-globstep-to-glob} we complete the lemma as follows:
  \begin{align*}
    \globstep(r,\sigma,r[p])[q_1]
    &= \glob(r[p]\sigma[q_1]) \\
    &= \glob(r[p]\sigma[q_2])
    = \globstep(r,\sigma,r[p])[q_2]
  \end{align*}
  The middle step in this reasoning chain is justified using
  Observation~\ref{obs:locally-closed-closed} and~\ref{obs:lift-subst} by
  \[r[p]\sigma[q_1] = \lift r[p]\sigma<q_1> = \lift r[p]<q_1>\sigma = \lift
    r[p]<q_2>\sigma = \lift r[p]\sigma<q_2> = r[p]\sigma[q_2].\qedhere\]
\end{proof}

Although the hypotheses in the lemma above are similar to the preconditions of a
single fork, there are some additional assumptions. This includes the existence of a
substitution list $\sigma$ and an assumption that $p \in \SCC(r)$. The
conclusion is also slightly off. Eventually, we need to conclude with a
statement of the form
\[\globstep(r,[],r)[pq_1] = \globstep(r,[],r)[pq_2],\]
where we have an empty substitution list, and the index $p$ is outside of
$\globstep$. The following technical lemma shows that if we perform enough
reduction steps of the globalization algorithm, this equation will eventually
take a shape suitable for Lemma~\ref{lem:fork0-to-alg}.

\begin{lemma}
  \label{lem:glob-boring-reduction}
  Let $r$ be closed term, $\sigma$ a substitution list, $p \in \SCC(r)$ and
  $q \in \valid{r[p]}$.
  Then there exist $\sigma_1$, $\sigma_2$, $r'$, $p'$ and $q'$ such that
  \begin{mathpar}
    pq = p'q' \and
    r' = r[p']\sigma_1 \and
    r' \text{ is closed} \and
    r'[q']\sigma_2 \text{ is closed} \and
    q' \in \SCC(r') \and
    \globstep(r, \sigma, r[p])[q] = \globstep(r', \sigma_2, r'[q']).
  \end{mathpar}
\end{lemma}
\begin{proof}
  Follows trivially by induction on $q$.
\end{proof}

Now follows the main fact that the existence of single fork between term nodes
means that their subjects are equal after globalization. The final theorem then
follows trivially from this.

\begin{lemma}
  \label{lem:same-fork-to-alg}
  If $\tn t_1[[p_1]] \sfork \tn t_2[[p_2]]$, then $\glob(t_1)[p_1] = \glob(t_2)[p_2]$.
\end{lemma}
\begin{proof}
  Recall that the single fork $\tn t_1[[p_1]] \sfork \tn t_2[[p_2]]$ can be
  built using two rules. We will provide a separate proof for each rule.

  \[
    \prftree[r]{closed}
    {t_1[p_1] \text{ closed}}
    {t_1[p_1] = t_2[p_2]}
    {\tn t_1[[p_1r]] \sfork \tn t_2[[p_2r]]}
  \]
  For this rule, the conclusion reduces to
  $\glob(t_1)[p_1r] = \glob(t_2)[p_2r]$. Note that it is sufficient to prove
  $\glob(t_1)[p_1] = \glob(t_2)[p_2]$. This follows directly from
  Lemma~\ref{lem:closed-to-alg} and the assumption $t_1[p_1] = t_2[p_2]$.

  \[
    \prftree[r]{\textit{let}-abs}
    {q_1 \text{ locally closed in } t[p]}
    {\lift t[p]<q_1> = \lift t[p]<q_2>}
    {\tn t[[pq_1r]] \sfork \tn t[[pq_2r]]}
  \]
  For this rule, the conclusion reduces to $\glob(t)[pq_1] = \glob(t)[pq_2]$.
  Because $t$ is closed, this can be expanded into $\globstep(t, [], t)[pq_1] =
  \globstep(t, [], t)[pq_2]$. We then use Lemma~\ref{lem:glob-boring-reduction}
  to obtain $r$, $\sigma_1$, $\sigma_2$, $p_1$ and $p_2$ such that
  \begin{mathpar}
    p = p_1 p_2 \and
    r = t[p_1]\sigma_1 \and
    r \text{ is closed} \and
    r[p_2]\sigma_2 \text{ is closed} \and
    p_2 \in \SCC(r) \and
    \globstep(t, [], t)[p] = \globstep(r, \sigma_2, r[p_2]).
  \end{mathpar}
  Note that from Observation~\ref{obs:lift-subst} we have that $q_1$ is locally
  closed in $r[p_2]$ and $\lift r[p_2]<q_1> = \lift r[p_2]<q_2>$.
  Lemma~\ref{lem:fork0-to-alg} then completes the proof by showing
  \[\globstep(r, \sigma_2, r[p_2])[q_1] = \globstep(r, \sigma_2, r[p_2])[q_2].\qedhere\]
\end{proof}

The hardest part of proving Theorem~\ref{thm:fork-to-alg} is already proven as
Lemma~\ref{lem:same-fork-to-alg},
now we finish it by considering arbitrary fork-equivalent pair instead of a single fork.

\ForkToAlgThm*
\begin{proof}
  A fork consists of a sequence of single forks. Each part of the sequence is
  proven correct by Lemma~\ref{lem:same-fork-to-alg}. The correctness of the
  complete sequence follows from the transitivity of Leibniz equality.
\end{proof}

\subsection{Algorithm implies Bisimilarity}
\label{sec:alg-to-bisim}

The main task in this section is to show that
\[
  \tn r[[\rootpos]] \bisim \tn \glob(r)[[\rootpos]].
\]
The main theorem then readily follows. To show this, we must first define the
bisimulation relation on $g$-terms. The transitions defined in
Definition~\ref{def:node-transitions} for $\lambda$-term nodes are lifted
verbatim to $g$-term nodes. Additionally, we extend the transition system by
adding appropriate outgoing edges to global variables. In particular, for any
$g$-term node $\tn t_1[[p_1]]$ such that $t_1[p_1]$ is of the form $\gvar{t_2\
  t_2[p_2]}$ and $t_2[p_2] \not \in \text{duplicates}(t_2)$ we have
\[
\tn t_1[[p_1]] \gstep{\varup} \tn t_2[[p_2]].
\]

\begin{remark}
  This definition of extra $\varup$ edges above is made specifically to match the efficient
  globalization algorithm.
  To reason about $\globnaive$, we would instead add a transition
  \[
    \tn t_1[[p_1]] \gstep{\varup} \tn t_2[[\rootpos]]
  \]
  for all $g$-term nodes such that $\tn t_1[[p_1]] = \gvar{t_2}$.
\end{remark}

\begin{remark}
  Since we added extra transitions, bimisilarity on $g$-terms does not match
  fork-equivalence on $g$-terms as it does for $\lambda$-terms. For example,
  the following two term nodes are bisimilar but not fork-equivalent.
  \[
    \tn(\lambda\ \debruijn 0)[[\lamdown]] \bisim
    \tn g((\lambda\ \debruijn 0)\ (\lambda\ \debruijn 0))[[\rootpos]]
  \]
\end{remark}

Having defined a suitable transition system to use for the bisimulation
relation, we can now start working towards a proof.
We start by observing a few technical facts about bisimilarity.
The following two observations are variants of
Lemma~\ref{lem:bisim-closed} and Lemma~\ref{lem:bisim-subpath}, trivially lifted
from $\lambda$-terms to $g$-terms.
\begin{observation}
  \label{obs:bisim-closed}
  Let $t$ and $t[p]$ be closed $g$-terms, then
  $\tn t[[p]] \bisim \tn t[p][[\rootpos]]$.
\end{observation}

\begin{observation}
  \label{obs:bisim-subpath2}
  If $q \in \valid{t_1[p_1]}$ and $\tn t_1[[p_1]] \bisim \tn t_2[[p_2]]$ then $\tn t_1[[p_1q]] \bisim \tn t_2[[p_2q]]$.
\end{observation}

And the following fact will be useful to prove bisimilarity by induction.
\begin{observation} Bisimulation between to $g$-term nodes can be established
  if their corresponding subterms are known to be bisimilar.
  \label{obs:bisim-recurse}
  \begin{align*}
    \tn t_1[[p_1\appleft]] \bisim \tn t_2[[p_2\appleft]] \wedge
    \tn t_1[[p_1\appright]] \bisim \tn t_2[[p_2\appright]] &\implies
    \tn t_1[[p_1]] \bisim \tn t_2[[p_2]] \\
    \tn t_1[[p_1\lamdown]] \bisim \tn t_2[[p_2\lamdown]] &\implies
    \tn t_1[[p_1]] \bisim \tn t_2[[p_2]]
  \end{align*}
\end{observation}

In the following two technical lemmata, we prove that a term is bisimilar to
itself, even if some of its de Bruijn indices have been replaced by appropriate
global variables.
\begin{lemma}
  \label{lem:subst1-bisim}
  Let $t$ be a closed $g$-term, $h$ be a $g$-term, and $n$ a $g$-term node such that
  \[
  \gvar{h} \gstep{\varup} n, \quad n \bisim \tn (\lambda\ t)[[\rootpos]].
  \]
  Then
  \[
    \tn t[\debruijn 0 \coloneqq \gvar{h}][[\rootpos]] \bisim \tn (\lambda\ t)[[\lamdown]]
  \]
\end{lemma}
\begin{proof}
  The bisimulation $n \bisim \tn (\lambda\ t)[[\rootpos]]$ gives us a
  bisimulation relation $R$. From that, we construct a new relation $R'$ as
  follows.
  \[ R' = R \cup \{ (\tn t[\debruijn 0 \coloneqq \gvar{h}][[p]], \tn (\lambda\ t)[[\lamdown p]]) \mid p \in \valid{t} \}\]
  It is easy to check that this relation is a bisimulation,
  therefore $R$ is included in bisimilarity, and
  $\tn t[\debruijn 0 \coloneqq \gvar{h}][[\rootpos]] \bisim \tn (\lambda\ t)[[\lamdown]]$.
\end{proof}

\begin{lemma}
  \label{lem:substs-bisim}
  Let $\sigma$ be a substitution into global variables, $h, t$ be $g$-terms,
  and $n$ be a $g$-term node such that
  \[
  \gvar{h} \gstep{\varup} n,\quad
  n \bisim \tn (\lambda\ t)\sigma [[\rootpos]].
  \]
  Then
  \[
    \tn t (\gvar{h}:\sigma) [[\rootpos]] \bisim
    \tn (\lambda\ t) \sigma [[\lamdown]]
    .
  \]
\end{lemma}
\begin{proof}
  Let $t' = (\lambda\ t)\sigma [\lamdown]$. Then
  $t (\gvar{h}:\sigma) = t'[\debruijn 0 \coloneqq \gvar{h}]$,
  and we obtain the result by applying Lemma~\ref{lem:subst1-bisim} to $t'$.
\end{proof}

Now we have all the tools needed to prove the key fact --
that the globalization algorithm does not change the term modulo bisimilarity.
We cannot prove this directly for the $\glob$ function. The induction hypothesis
would be too weak. Instead we prove a stronger statement for $\globscc$.

\begin{lemma}
  \label{lem:globalize-bisim2}
  let $r$ be a closed $g$-term, $p\in\SCC(r)$,
  and let $\sigma$ be a list of global variables. Assume
  \begin{equation}
    \label{eqn:bisim-assump}
    \tn r[[p]] \bisim \tn r[p]\sigma [[\rootpos]],
  \end{equation}
  and $r[p] \not\in \text{duplicates}(r)$.
  Then we obtain
  \begin{equation}
    \label{eqn:bisim-general}
    \tn r[[p]] \bisim \tn \globscc(r,\sigma,r[p])[[\rootpos]].
  \end{equation}
\end{lemma}
\begin{proof}
  Proceed by induction on $r[p]$.
  \begin{description}
  \item[Case {$r[p] = \debruijn i$}:]
    We have $\globscc(r,\sigma,r[p]) = \debruijn{i}\sigma = r[p]\sigma$.
    The conclusion the follows directly from Equation~\ref{eqn:bisim-assump}.
  \item[Case {$r[p] = \gvar{h}$}:]
    We have $\globscc(r,\sigma,r[p]) = \gvar{h} = r[p]\sigma$. The conclusion
    again follows from Equation~\ref{eqn:bisim-assump}.
  \item[Case {$r[p] = \lambda\ t$}:]
    We have
    \begin{equation}
      \label{eqn:globalize-lambda-reduction}
      \globscc(r,\sigma,r[p]) = \lambda\ \globstep(r,(\gvar{r\ r[p]} :
      \sigma),r[p\lamdown]).
    \end{equation}
    Further reduction of the algorithm depends on whether $\globstep$ transitions
    to $\glob$ or $\globscc$. We will consider both cases separately. However,
    in both cases we will require the following fact:
    \begin{equation}
      \label{eqn:subterm-bisim}
      \tn r[[p\lamdown]] \bisim \tn r[p\lamdown](\gvar{r\ r[p]} : \sigma) [[\rootpos]].
    \end{equation}
    This follows from Lemma~\ref{lem:substs-bisim}, instantiating node $n$ to
    $\tn r[[p]]$.

    Now, Equation~\ref{eqn:globalize-lambda-reduction} can reduce further
    according to two possibilities:
    \begin{enumerate}
    \item If
      $r[p\lamdown]$ is closed or $r[p\lamdown] \in \text{duplicates}(r)$, it
      reduces to
      \[\lambda\ \glob(r') =
        \lambda\ \globscc(r', [], r'),\]
      where $r' = r[p\lamdown](\gvar{r\ r[p]} : \sigma)$. From the induction
      hypothesis, using $\tn r'[[\rootpos]] \bisim \tn r'[\rootpos][][[\rootpos]]$, we then obtain
      \[\tn r'[[\rootpos]] \bisim \tn \globscc(r', [], r')[[\rootpos]].\]
      Combining this with with Equation~\ref{eqn:subterm-bisim} yields
      \[\tn r[[p\lamdown]] \bisim \tn \globscc(r', [], r')[[\rootpos]].\]
      The final conclusion then follows from
      Observation~\ref{obs:bisim-recurse}.
    \item Otherwise, if $r[p\lamdown]$ is neither closed nor duplicate, then
      Equation~\ref{eqn:globalize-lambda-reduction} reduces to
      \[\lambda\ \globscc(r,(\gvar{r\ r[p]} : \sigma),r[p\lamdown]).\]
      By applying Equation~\ref{eqn:subterm-bisim} on the induction hypothesis,
      we obtain
      \[\tn r[[p\lamdown]] \bisim
        \tn \globscc(r,(\gvar{r\ r[p]} : \sigma),r[p\lamdown])[[\rootpos]].\]
      Finally, the conclusion again follows from
      Observation~\ref{obs:bisim-recurse}.
    \end{enumerate}
  \item[Case {$r[p] = t\ u$}:]
    We have
    \[
      \globscc(r,\sigma,r[p]) =
      \globstep(r, \sigma, r[p\appleft])\ \globstep(r, \sigma, r[p\appright]).
    \]
    Similar to the two reduction options of the previous case, the two instances
    of $\globstep$ either reduce further to $\globscc$ or $\glob$. Using similar
    reasoning to the previous case, we can use the induction hypothesis together
    with
    \begin{mathpar}
      \tn r[[p\appleft]] \bisim \tn r[p\appleft]\sigma[[\rootpos]] \and
      \tn r[[p\appright]] \bisim \tn r[p\appright]\sigma[[\rootpos]] \and
    \end{mathpar}
    to obtain
    \begin{mathpar}
      \quad\tn r[[p\appleft]] \bisim \tn \globstep(r, \sigma, r[p\appleft])[[\rootpos]] \and
      \tn r[[p\appright]] \bisim \tn \globstep(r, \sigma, r[p\appright])[[\rootpos]].
    \end{mathpar}
    Finally, the conclusion again follows from
    Observation~\ref{obs:bisim-recurse}.\qedhere
  \end{description}
\end{proof}

\begin{corollary}
  \label{corr:globalize-bisim}
  \[
    \tn r[[\rootpos]] \bisim \tn \glob(r)[[\rootpos]].
  \]
\end{corollary}
\begin{proof}
  Follows directly from Lemma~\ref{lem:globalize-bisim2}.
\end{proof}

Proving Theorem~\ref{thm:alg-to-bisim} is now straightforward.

\AlgToBisimThm*
\begin{proof}
  Since $\glob(t_1)$, and $\glob(t_1)$ are closed $g$-terms,
  this result is obtained from the following equivalence chain provided by
  Lemma~\ref{corr:globalize-bisim}, Observation~\ref{obs:bisim-subpath2}, and Observation~\ref{obs:bisim-closed}.
  \begin{align*}
  \tn t_1[[p_1]]
  \bisim \tn \glob(t_1)[[p_1]]
  \bisim\ &\tn \glob(t_1)[p_1][[\rootpos]]\\
  =\ &\tn \glob(t_2)[p_2][[\rootpos]]
  \bisim \tn \glob(t_2)[[p_2]]
  \bisim \tn t_1[[p_1]]
  \end{align*}
\end{proof}

\end{document}